\documentclass[letterpaper,onecolumn,12pt,unpublished]{quantumarticle} 
\pdfoutput=1
\usepackage[utf8]{inputenc}
\usepackage[english]{babel}
\usepackage[T1]{fontenc}
\usepackage{hyperref}

\usepackage{amsmath,amsthm,amsfonts,amssymb,braket}

\usepackage[all]{xy}
\usepackage{fullpage}
\usepackage{mathtools}
\usepackage[numbers]{natbib}
\usepackage[export]{adjustbox}

\usepackage[nameinlink,capitalize]{cleveref}

\usepackage{subcaption}

\numberwithin{equation}{section}
\theoremstyle{plain}
\newtheorem{lemma}[equation]{Lemma}

\newtheorem{proposition}[equation]{Proposition}
\newtheorem{proposition-definition}[equation]{Proposition-Definition}
\theoremstyle{definition}

\newtheorem{remark}[equation]{Remark}

\newcommand{\CC}{{\mathbb{C}}}
\newcommand{\RR}{{\mathbb{R}}}

\newcommand{\calA}{{\mathcal{A}}}

\newcommand{\calC}{{\mathcal{C}}}

\newcommand{\calJ}{{\mathcal{J}}}

\newcommand{\calN}{{\mathcal{N}}}

\newcommand{\one}{{\mathbf{1}}}
\newcommand{\rd}{{\mathrm{d}}}

\newcommand{\SWAP}{{\mathrm{SWAP}}}

\newcommand{\half}{{\tfrac 1 2}}

\usepackage{tikz}
\usepackage{lipsum}

\newcommand{\U}{{\mathsf{U}}}

\newcommand{\proj}[1]{\ket{#1}\!\bra{#1}}

\DeclareMathOperator{\poly}{poly}
\DeclareMathOperator{\polylog}{polylog}

\newcommand{\ii}{{\mathbf{i}}}

\DeclareMathOperator{\tr}{{{Tr}}}
\DeclareMathOperator{\Tr}{{Tr}}

\DeclareMathOperator*{\diag}{{diag}}

\DeclareMathOperator{\rank}{{rank}}

\DeclareMathOperator{\avg}{{\mathbb{E}}}

\DeclarePairedDelimiter{\norm}{\lVert}{\rVert}
\DeclarePairedDelimiter{\abs}{|}{|}

\begin{document}

\title{Growth and collapse of subsystem complexity under random unitary circuits}

\author{Jeongwan Haah}
\affiliation{Leinweber Institute for Theoretical Physics, Stanford University, California, USA}
\affiliation{Google Quantum AI}
\orcid{0000-0002-1087-6853}
\author{Douglas Stanford}
\affiliation{Leinweber Institute for Theoretical Physics, Stanford University, California, USA}

\begin{abstract}
    For chaotic quantum dynamics modeled by random unitary circuits,
    we study the complexity of reduced density matrices of subsystems
    as a function of evolution time where the initial global state is a product pure state.
    The state complexity is defined as the minimum number of local quantum channels
    to generate a given state from a product state to a good approximation.
    In $1+1$d, we prove that 
    the complexity of subsystems of length~$\ell$ smaller than half 
    grows linearly in time~$T$ at least up to~$T = \ell / 4$ 
    but becomes zero after time~$T = \ell /2$ 
    in the limit of a large local dimension,
    while the complexity of the complementary subsystem of length larger than half
    grows linearly in time up to exponentially late times.
    Using holographic correspondence,
    we give some evidence that the state complexity of the smaller subsystem
    should actually grow linearly up to time $T = \ell/2$ and then abruptly decay to zero.
\end{abstract}

\maketitle
\tableofcontents

\section{Introduction}

The complexity of a pure quantum state $\mathfrak{C}(|\psi\rangle)$ is defined as the minimum number of simple unitary operations needed to approximately prepare it, starting from a product state.%
\footnote{There is more than one reasonable definition of the state complexity. See \cref{rem:complexityDef}.}
A fundamental problem that has been studied from several points of view is the following: how does the complexity of 
\begin{equation}
|\psi(T)\rangle = e^{-\mathrm{i}H T}|\psi_0\rangle, \hspace{20pt} (|\psi_0\rangle = \text{ simple initial state})
\end{equation}
depend on time? An obvious recipe to prepare $|\psi(T)\rangle$ is just to (approximately) implement the Hamiltonian evolution itself. This gives an upper bound $\mathfrak{C} \le (\text{const})\cdot T$, and the question is whether this is close to saturated. Holographic conjectures \cite{Susskind:2014rva,Brown:2017jil}, simple model systems \cite{Lin:2018cbk}, and work on random quantum circuits \cite{Brandao:2012zoj,Roberts:2016hpo,Jian:2022pvj,Oszmaniec:2022srs,Li:2022hkh,Haferkamp:2023cem,Chen2024} suggest that if the time evolution is sufficiently chaotic, no shortcut or compression is possible until late times, exponential in the entropy.

One can also define the complexity of a mixed state $\sigma$, as the minimum number $N = \mathfrak C(\sigma)$ of local quantum channels~$\calC_i$ (completely positive trace-preserving linear maps on density matrices) to create a good approximation to~$\sigma$ starting with a product state:
\begin{equation}\label{defofC}
    \frac 1 2 \norm[\Big]{ 
        \sigma - \calC_N \circ \cdots \circ \calC_2 \circ \calC_1 ( \eta_1 \otimes \cdots \otimes \eta_\ell )
    }_1
    \le \alpha.
\end{equation}
Here $\eta_j$ is the initial density matrix at qudit~$j$ 
and each~$\calC_i$ is chosen from a set of allowed elementary operations.
We allow ancillary qudits inserted and discarded in between the local channels.
These ancillas are all qudits of the same dimension as those in the system.
In a robust definition of complexity,
we should allow a constant error like $\alpha = 0.1$.

If a chaotic system is in state $|\psi(T)\rangle$, subsystems will be in mixed states. How does the complexity of one of these mixed states depend on $T$? Before describing our results, we will discuss some reasonable expectations.
\begin{itemize}
\item An upper bound is given by a preparation recipe where we first prepare the overall pure state $|\psi(T)\rangle$ and then trace out the complementary subsystem. For local systems, this would give $\mathfrak{C} \le (\text{const.})LT$ where $L$ is the size of the total system. 

\item
For short times, one expects that for any subsystem, the complexity will grow at least somewhat, although perhaps slower than $\propto LT$ because it may be unnecessary to implement the time evolution fully on the complementary subsystem. 

\item
For late times, it seems important to distinguish between subsystems that are smaller than half vs.~bigger than half. 
\begin{itemize}
\item Subsystems that are smaller than half are expected to thermalize, so their late-time density matrix approaches a thermal density matrix, with low complexity (zero, at infinite temperature). 
\item If the subsystem is larger than half, then its density matrix will continue to evolve \cite{Tang:2024puq}, and one might expect the complexity to grow for a long time.
\end{itemize}
\end{itemize}
An interesting question is how the short time and long time behaviors are connected together. For example, how does the complexity of the smaller system approach its small late-time value? One could imagine either a continuous or (approximately) discontinuous evolution
as in \Cref{fig:functionCt}.
\begin{figure}[b]
\begin{equation*}
\begin{tikzpicture}[scale=2, rotate=0, baseline={([yshift=-0.15cm]current bounding box.center)},line cap=round, line join=round, thick, >=latex]

  \draw[->, gray] (-0.1,0) -- (2.1,0) node[below right] {$T$};
  \draw[->, gray] (0,-0.1) -- (0,1.2) node[left] {$\mathfrak{C}$};
\draw[black, very thick] (1,0) -- (2,0);
  \draw[domain=0:1, smooth, variable=\x, very thick] plot ({\x}, {2*\x*(1-\x)});
\end{tikzpicture}
\hspace{20pt}
\begin{tikzpicture}[scale=2, rotate=0, baseline={([yshift=-0.15cm]current bounding box.center)},line cap=round, line join=round, thick, >=latex]

  \draw[->, gray] (-0.1,0) -- (2.1,0) node[below right] {$T$};
  \draw[->, gray] (0,-0.1) -- (0,1.2) node[left] {$\mathfrak{C}$};

  \draw[black, very thick] (0,0) -- (1,1) -- (1,0) -- (2,0);  
\end{tikzpicture}
\end{equation*}
\caption{Two possible behaviors of the complexity~$\mathfrak C$ of a smaller subsystem as a function of evolution time~$T$ }
\label{fig:functionCt}
\end{figure}
Conjectures for this function $\mathfrak{C}(T)$ can be obtained by studying holographic proposals for the complexity of subsystems. According to these proposals, either qualitative scenario can happen \cite{Chen:2018mcc} (corresponding to continuous vs discontinuous saturation scenarios of \cite{Liu:2013iza}), but for the particular case of a $1+1$-dimensional system starting in a state with short-ranged entanglement, the holographic prediction looks like the figure on the right.

{\bf In this paper}, we will try to put the expectations listed above on a firmer footing, working in the setting of random brickwork quantum circuits, with $1+1$ dimensional architecture. In particular, using rigorous methods we establish
\begin{itemize}
    \item In \cref{sec:ComplexityGrowth}: the complexity of a subsystem bigger than half has a $T$-linear lower bound until exponential times.
    \item In \cref{secthree}: systems smaller than half thermalize to infinite temperature and therefore have small late time complexity.
    \item In \cref{secmutual}: the behavior of mutual information implies that the complexity of a smaller subsystem has a lower bound of the rough form in the left panel of \Cref{fig:functionCt}.
\end{itemize}
These results do not distinguish between the scenarios in \Cref{fig:functionCt}.
So in the remainder of the paper we give some non-rigorous evidence in favor of the discontinuous one. In particular
\begin{itemize}
    \item In \cref{secholographic} we review holographic predictions for complexity and use entanglement wedge reconstruction to give a perspective on the sharp transition.
    \item In \cref{secreplica} we imitate the entanglement wedge reconstruction discussion for a variant of the random quantum circuit. We show that up until right before thermalization, the small subsystem remembers all of the gates that were applied to it in the past.
\end{itemize}
As a separate question (motivated by the memory capacity of the almost-thermalized small subsystem), we study in \cref{counting} the problem of how many approximately pairwise orthogonal density matrices can exist, as a function of their rank and dimension.

\subsection*{Brickwork random circuit}
Let's now introduce the brickwork circuit that we will study in most of the paper. We consider random circuit dynamics on systems of $q$-dimensional qudits.
Initially the system is in a product state,
and at every time step we apply a layer of two-site gate
chosen uniformly at random from~$\U(q^2)$ in a brickwork fashion.
If the system is a circle containing $L$ sites,
then our brickwork circuit requires that $L$ be even.
Let there be $T$ layers of random gates in total.
We will consider intervals,
and to avoid undesired triviality of the last time step,
we assume that both boundary points of our interval are straddled by a gate at time~$T$.
This implies that the length~$\ell$ of our interval is always even.

\begin{remark}\label{rem:complexityDef}
In computational complexity,
there is no universal definition for ``elementary operations.''
This ambiguity continues to exist in our study:
we choose to allow elementary quantum channels to use ancillas,
whereas other considerations may not;
we take arbitrary two-qudit operations as elementary,
while in some other cases some specific choices are preferred;
we allow a constant error measured in trace distance in the approximation,
while others may require smaller (more stringent) error.

Nevertheless, we can often meaningfully compare 
the complexity of functions or distributions without those details
because they only introduce minor overall factors.
This is well established for different generating sets for a fixed unitary group
by the Solovay--Kitaev theorem,
though the ambiguity with ancillas appears to be more involved.
We do not address these differences in this paper,
but we note two important choices we make.

One is locality.
Our elementary quantum channels always act on two \emph{neighboring} qudits.
This is physically well motivated, 
and contrasts to a more common setting in computational complexity literature
where any two-qudit gate is deemed elementary,
no matter how far the two qudits are.

The other is related to our interest in which the local dimension $q$ is large.
For large~$q$, a typical Haar-random $\U(q^2)$ gate 
would have high complexity from the perspective of elementary qubit operations. 
One could try to replace these random unitaries 
with elements of a unitary design with smaller complexity. 
However, we will simply declare that any $\U(q^2)$ operation acting on neighboring qudits
is an allowed~$\calC_i$ and has complexity one by definition.
\end{remark}

{\bf Note:} we became aware of related work~\cite{Fan2025} and we are coordinating the submission of our papers.

\section{Complexity growth of a larger region}\label{sec:ComplexityGrowth}

Let $\sigma$ be a density operator obtained by tracing out a subsystem of $L - \ell$ qudits 
from a pure state on~$L$ qudits.
So, $\sigma$ is supported on $\ell$ qudits.
We consider the probability that 
the random circuit produces two states~$\sigma$ there are close to each other.
The case where $\ell = L$ was considered in~\cite{Chen2024}.

We find it convenient to use fidelity 
$F(\sigma,\sigma') = \Tr\sqrt{\sqrt{\sigma}\sigma' \sqrt{\sigma}} = \norm{\sqrt{\sigma}\sqrt{\sigma'}}_1$
as distance measure where $\norm{A}_1 = \Tr\sqrt{A^\dagger A}$.
Fidelity,\footnote{Some authors define fidelity as the square of ours.}
trace distance $T(\sigma,\sigma') = \half \norm{\sigma - \sigma'}_1 
= \half \Tr \sqrt{(\sigma - \sigma')^2}$,
and Hilbert--Schmidt distance $E(\sigma,\sigma') = \half \norm{\sigma - \sigma'}_2 = \half \sqrt{\Tr((\sigma - \sigma')^2)} $ are related by
the Fuchs--van de Graaf and Schatten norm inequalities
\begin{align}
    &F^2 + T^2 \le 1 \le T + F \, ,\\
    &E \le T \le 2 E \sqrt{\min(\rank \sigma, \rank \sigma')}\, .\nonumber
\end{align}

Our probability estimate will use a two-step argument.
The first step is that random circuits generate 
an $\epsilon$-approximate unitary $k$-design~$\nu$ 
(a probability distribution on $\U(q^L)$)
in depth linear in~$k$.
The defining inequality\footnote{The usual definition of relative error approximate unitary designs 
    has the $\nu$-average sandwiched by the Haar average up to factors~$1 \pm \epsilon$,
    but we need one direction as shown.
} 
for the approximate unitary design~$\nu$ is as follows:
\begin{align}
  \avg_{U \sim \nu} \Tr(M U^{\otimes k} \Sigma U^{\dagger \otimes k} ) 
  \le (1+\epsilon)
  \avg_{V \sim \U(q^L)} \Tr(M V^{\otimes k} \Sigma V^{\dagger \otimes k} ) 
  \label{eq:approxDesignEps}
\end{align}
for all \emph{positive} semidefinite operators~$M$ and $\Sigma$ 
acting on $q^{k L}$-dimensional Hilbert space,
where $V$ is Haar random.
The random brickwork circuit~$U$ of depth~$T$ is an $\epsilon$-approximate unitary $k$-design 
with $\epsilon < 1$ and $k = T / \poly(L)$ if $T \le q^{c L}$ for some constant~$c > 0$~\cite{Chen2024}.%
\footnote{Strictly speaking, to directly quote the result of~\cite{Chen2024}, 
    we need to assume that $q$ is a power of~$2$.
    If $q$ is a power of~$2$, our brickwork dynamics 
    using $q$-dimensional qudits approaches 
    the Haar random one only faster than the $q=2$ case,
    which can be seen by padding a depth-one circuit with $\U(q^2)$-gates before and after
    a depth-2 circuit with $\U(4)$-gates.
}

The second step converts~\eqref{eq:approxDesignEps} into the following probability estimate,
derived below.
Namely, for \emph{any} density matrix~$\rho$ on the same subsystem as~$\sigma$, 
we have
\begin{align}
  \Pr_\sigma \big[ F(\rho,\sigma) \ge 1 - 2\alpha \big] \le  \frac{\avg_\sigma F(\rho,\sigma)^{2k}}{(1-2\alpha)^{2k}} \le \frac{(1+\epsilon) }{(1-2\alpha)^{2k} q^{k(2\ell - L)}} \min(e^{k^2 q^{-(L-\ell)}},~ k!) \, .
  \label{eq:ProbLargeOverlapInLargerInterval}
\end{align}
This second step does not use any detail of the dynamics but only~\eqref{eq:approxDesignEps}.
Setting $k = T / \poly(L)$ and $\alpha = 0.1$, and assuming~$\ell > L /2$,
we conclude that the probability of interest is at most
$q^{ - T / \poly(L)}$
as long as $T \le \min(\poly(L) \max(q^{2\ell - L}, q^{(L-\ell)/2}),\, q^{cL})$.
Note that $\rho$ is completely arbitrary.

Now, let $\mathcal S$ be the set of all states that are expressible by $G$ gates.
Suppose we have a set~$\{\rho_1, \rho_2,\ldots, \rho_N \} \subset \mathcal S$
such that for every $\rho \in \mathcal S$ there exists $\rho_i$ 
with $\frac 1 2 \norm{\rho - \rho_i}_1 \le \alpha = 0.1$.
Such a set is called an $\alpha$-covering net for~$\mathcal S$.
If $\sigma$ is $\alpha$-close in trace distance to some element of~$\mathcal S$,
then $\sigma$ is $2\alpha$-close in trace distance to some~$\rho_i$.
Hence, the probability that our reduced density matrix~$\sigma$ generated by the random circuit
has complexity~$G$ is at most $N q^{-T / \poly(L)}$.
We want this probability to be so small that the complexity of~$\sigma$ is typically large.
Our interest is thus an upper bound on~$N$.

If we perturb each gate in a circuit of $G$ gates 
by $\frac \alpha {G}$ in the channel metric (diamond norm),
then the final state will differ in trace distance by at most~$\alpha$.
This means that for the construction of~$\{\rho_1,\ldots,\rho_N\}$, 
it is extravagant to consider two gates at a location that differ 
by~$\frac \alpha {2G}$ or smaller in the channel metric.
It is not difficult to construct a $\frac \alpha {G}$-covering net of 
cardinality~$(G/\alpha)^{\poly(q)}$
for the space of all two-qudit channels, 
where the exponent~$\poly(q)$ is the dimension.%
\footnote{This covering net for the space of local channels is not needed 
    if we declare that only $\poly(q)$ choices of two-qudit channels are elementary,
    in which case $N \le \poly(q,L)^G$.
}
We conclude that $N \le \poly(q,L, (G/\alpha)^{\poly(q)})^G$.

It follows that most reduced density matrices~$\sigma$ produced by the random circuit
requires $G$ local channels where
\begin{equation}
    G \ge T/\poly(q, L,\log T) \ge T / \poly(qL) \quad \text{ if } \quad T \le q^{c' L} 
\end{equation}
for some universal constant~$c' > 0$,
where the second inequality is because $\log T \le c' L \log q$.
This shows the linear-$T$ growth of the complexity of most instances~$\sigma$
until exponentially late time.

In short, the argument is that the random circuits produce $N \ge q^{T / \poly(q L)}$ 
nearly orthogonal states on the subsystem of size~$\ell$ 
so we need $\log(N) \ge T / \poly(q L)$ gates to express them.
Note that the result here,
especially \eqref{eq:ProbLargeOverlapInLargerInterval}, 
is meaningful only for any subsystem of size~$\ell$ such that
\begin{equation}
    \ell > L / 2 \, , \label{eq:factorOfTwo}
\end{equation}
but otherwise the subsystem is arbitrary;
the subsystem does not have to be contiguous.

\begin{remark}
    The factor of~$2$ in the condition \eqref{eq:factorOfTwo} appears in the proof
    at a seemingly minor step where we bound the rank of~$\sigma$ by~$q^{L-\ell}$.
    The condition~$\ell > L /2$ is just to ensure that $q^{L-\ell} < q^{\ell}$.
    Indeed, if the rank of~$\sigma$ is much smaller than~$q^\ell$ (the maximum possible),
    still the argument goes through and the linear growth of subsystem complexity holds.
    This leads to very different phenomenology of complexity 
    for other random circuits,
    for example, the ``shallow'' unitary designs~\cite{LaRacuente2024,Schuster2024,Cui2025},
    which we explain in~\cref{app:shallowDesigns}.
\end{remark}

\subsection*{Proof of \cref{eq:ProbLargeOverlapInLargerInterval}}
The first inequality in~\cref{eq:ProbLargeOverlapInLargerInterval} is Markov's inequality.
The second inequality is proved as follows.
The fidelity between $\rho$ and $\sigma$ has its $k$-th moment bounded as
\begin{align}
  F(\rho,\sigma)^{2} 
  &= 
  \norm{\sqrt \rho \sqrt \sigma}_1^2 \nonumber\\
  &\le 
  r \norm{\sqrt \rho \sqrt \sigma}_2^2 & (r = \rank \sigma) \label{eq:FtoTwonorm}\\
  &=
  r \Tr(\rho \sigma) \nonumber \\
  &=
  r \bra \psi (\underbrace{\rho \otimes \one_{r'}}_{\rho' \succeq 0}) \ket \psi & (r' = q^{L-\ell}, \text{$\psi$ purifies $\sigma$}) \nonumber
\end{align}
\begin{align}
  \avg_\sigma F(\rho,\sigma)^{2k} 
  &\le 
  r^k \avg_\psi \bra{\psi} \rho' \ket{\psi}^k & \text{(by \eqref{eq:FtoTwonorm})}\nonumber\\
  & =
  r^k \avg_U \bra{0^{\otimes k}} U^{\dagger \otimes k} (\rho')^{\otimes k} U^{\otimes k} \ket{0^{\otimes k}} 
    & (\ket \psi = U \ket 0 )\nonumber\\
  &\le
  r^k (1+\epsilon) \avg_V \bra{0^{\otimes k}} V^{\dagger \otimes k} (\rho')^{\otimes k} V^{\otimes k} \ket{0^{\otimes k}} 
    & (\text{\eqref{eq:approxDesignEps} with }(\rho')^{\otimes k} \succeq 0)\nonumber\\
  &=
  r^k (1+\epsilon) \avg_V (\Tr V \rho' V^\dagger Q )^k & (Q = \proj{0})\nonumber\\
  &\le
  r^k (1+\epsilon) \avg_V \avg_i (\Tr V P_i V^\dagger Q )^k & (\rho' = \avg_i P_i,~\rank P_i = r')\nonumber\\
  &\le
  (1+\epsilon) \left(\frac{r r'}{d}\right)^k \,\, \prod_{j=0}^{k-1}\left(1 + \frac{j}{r'}\right) & (d = q^L)
       \, \label{eq:AvgFidelityKBound}
\end{align}
where the second to the last inequality uses the convexity of $\phi : y \mapsto y^k$,
and the last inequality uses~\cref{thm:randomProjectorOverlap} below.
We finally use $r \le q^{L-\ell} = r'$  to arrive at the second inequality 
of~\cref{eq:ProbLargeOverlapInLargerInterval}.

\begin{lemma}\label{thm:randomProjectorOverlap}
  Let $A$ and $B$ be orthogonal projectors of rank~$a$ and~$b$ on~$\CC^d$, respectively.
  Let $U \sim \U(d)$ be a Haar random unitary. 
  For any $i$, let $x_i \sim \mathcal N(0,\half)$ be an independent real Gaussian random variable.
  Then, for any convex function $\phi: \RR_{\ge 0} \to \RR$, we have
  \begin{align}
    \avg_{U} \phi\left[ \frac{d}{ab}\Tr( UPU^\dagger Q ) \right] 
    &\le
    \avg_{x_i} \phi\left[ \frac{1}{ab} \sum_{i=1}^{2ab} x_i^2 \right] \, ,\\
    \avg_{x_i} \left(\frac{1}{m} \sum_{i=1}^{2m} x_i^2 \right)^k 
    &= 
    \prod_{j=0}^{k-1}\left(1 + \frac j m \right)
    \, .
  \end{align}
\end{lemma}

\noindent
This is a minor generalization of~\cite[Lemma~III.5]{Hayden2004},
in which $x \mapsto \exp(\xi x)$ was used as a convex function with an optimized $\xi \in \RR$.
A self-contained short proof of~\cref{thm:randomProjectorOverlap} 
is presented in~\cref{app:randomProjLemma}.

\section{Late-time behavior of a smaller interval}\label{secthree}

We have shown that 
the complexity of a larger-than-half interval grows linearly in time up to an exponentially long time.
In contrast, an interval of length~$\ell < \frac L 2 ( 1 - \delta)$ eventually thermalizes,
where $\delta \propto \frac{\log 2}{\log q}$.
More precisely, we have the following, the proof of which is found in~\cref{app:Purity}.

\begin{proposition}\label{thm:InfinitePurityBound}
    Let $P(T;\ell)$ be the ensemble-averaged purity of a finite interval of length~$\ell$ at time~$T$
    under the brickwork random circuit dynamics starting with a product pure state on an infinite chain.
    For any~$q \ge 1$, integer~$T \ge 0$, and even integer~$\ell \ge 2$, we have
  \begin{equation}
    \frac{1-q^{-\ell}}{T+1}\binom{2T}{T}\left(\frac{q}{q^2 + 1} \right)^{2T} 
    \le 
    \left(
      P(T;\ell) -\frac{1}{q^{\ell}}
    \right) 
    \le 
    \left(\frac{2q}{q^2 + 1} \right)^{2T}\, .
  \end{equation}
\end{proposition}
\noindent 
Note that $\log \binom{2T}{T} = 2T \log 2 + \cdots$,
so the bounds are quite tight.

It follows from \cref{thm:InfinitePurityBound} that
if $\rho(U)$ is the reduced density operator obtained by tracing out 
all but $\ell$ consecutive qubits 
after applying a random unitary circuit~$U$ of depth~$T$ on a product state,
then
for any $q \ge 2$ and any $\alpha > 0$ we have
\begin{equation}
    \avg_U \norm*{\rho(U) - \frac{\one_{q^\ell}}{q^\ell} }_1 
    \le q^{\ell/2}\avg_U \norm*{\rho(U) - \frac{\one_{q^\ell}}{q^\ell} }_2 
    \le 
    \sqrt{q^\ell \avg_U P(T;\ell,U) - 1} < \alpha
\end{equation} 
if
\begin{equation}
    T > \frac \ell 2 \left(1 + 3\frac{\log 2}{\log q} \right) + \frac{3\log \tfrac 1 \alpha}{\log q} \, .
\end{equation}

To use this result on a circle of circumference~$L$,
we can examine an interval of length~$\ell < \frac{L}{2}( 1 - 6 \frac{\log 2}{\log q})$.
The constant~$6$ is unimportant.
For such an interval smaller than half, the reduced density matrix at time
\begin{equation}
    T_1 = \frac \ell 2 \left( 1 + 6 \frac{\log 2}{\log q} \right)
\end{equation}
does not depend on the boundary condition of the chain 
because the lightcone of the interval does not cover the whole circle,
so we can use the infinite-line result of~\cref{thm:InfinitePurityBound}.
This means that, with high probability, 
the state on a smaller-than-half interval on the circle 
is close to the maximally mixed state
and therefore has zero complexity.

Since our requirement on the length~$\ell$ of the subsystem
is that $\ell < \frac L 2 ( 1- \delta)$,
there remains a small uncharted regime 
for intervals whose length is very close to but less than~$L/2$.

\section{Complexity growth as seen by mutual information}\label{secmutual}

In short time regimes, say $T \lesssim \ell  /2$, 
the estimates for unitary design quality does not imply anything meaningful.%
\footnote{
    There is room for better analysis of 
    brickwork random unitary circuit,
    and one might be able to derive a complexity bound similar to ours in this section
    using a counting argument.
    Currently, the 1d brickwork random unitary circuit~$U$ 
    is known to become an $\epsilon$-approximate unitary $k$-design 
    at depth $T  = \tilde O( k L \log \frac 1 \epsilon)$~\cite{Chen2024}
    However, in view of~\cite{Schuster2024},
    it is conceivable that $U$ is already an $\epsilon$-approximate unitary $k$-design
    at depth $T = O(k \polylog(k L / \epsilon))$.
    After all, the result of~\cite{Schuster2024} 
    says that if we omit some random gates at specific locations from~$U$,
    then it is indeed the case. See~\cref{app:shallowDesigns}.
}
However, this was not a serious issue for larger-than-half intervals
since we have shown that the typical complexity grows uniformly linearly in time 
up to $T_{\mathrm{max}} = \exp(\Theta(\ell))$.
On the other hand, for smaller-than-half intervals
we only know so far that they typically become zero complexity for $T \gtrsim \ell / 2$.

In this section,
we derive a lower bound for typical complexity in the short time regime $T < \ell / 2$
for smaller-than-half intervals
by looking at mutual information profile,
which is an argument different from the counting in~\cref{sec:ComplexityGrowth}.
Our lower bound grows in time~$T$, peaks at $T \approx \ell / 4$, 
and then decays to zero at $T \approx \ell / 2$.
Therefore, the complexity of smaller-than-half intervals grows linearly in time
at least until~$T \approx \ell/4$. 
Note that our upper and lower bounds on the complexity of smaller-than-half interval
do not match, even qualitatively, in the regime $\ell / 4 \lesssim T \lesssim \ell / 2$.

\subsection{Mutual information growth}

Consider a bipartition $AB$ of the system in a state~$\rho$ 
and a quantum channel~$\calC$ acting on~$A$.
Let $\rho' = \calC(\rho)$.
By Stinespring, the channel~$\calC$ is an isometric embedding of~$A$ into $AE$ 
($E$ being an ancilla)
followed by the trace-out of~$E$.
\begin{equation}
    \rho = \rho^{AB} 
    \xrightarrow{\quad \text{iso} \quad } (\rho')^{ABE} 
    \xrightarrow{\quad \Tr_E \quad} \rho' = (\rho')^{AB} 
\end{equation}
The strong subadditivity of von Neumann entropy implies that
the mutual information of~$\rho$ is nonincreasing:
\begin{equation}
    I_\rho(A\!:\!B) = I_{\rho'}(AE:B) \ge I_{\rho'}(A\!:\!B) \, . \label{eq:MInonincreasing}
\end{equation}

Let $\sigma$ be a density matrix of interest, and
suppose that a circuit of local quantum channels generates a bipartite state~$\rho$ 
such that $\frac 1 2 \norm{\rho - \sigma}_1  = \alpha$
where there are $m$ local channels in the circuit 
that straddle between complementary subsystems~$A$ and~$B$.
Using Stinespring, 
we regard that all local channels are implemented by some local isometry
using ancillas that are going to be traced out.
The isometry in a dilation of a channel~$\calC$ can always be chosen to
be a matrix product operator of bond dimension that is the product of
Hilbert space dimensions of the input and output.
This gives an upper bound~$q^{2m}$ on the bond dimension across~$A$ and~$B$
for~$\rho$ as a tensor network.
In turn, the bond dimension bounds the von Neumann mutual information from above
by $4m \log q$.
Since tracing out does not increase the mutual information \eqref{eq:MInonincreasing},
we have
\begin{equation}
    I_{\rho}(A\!:\!B) \le 4 m \log q \, .
\end{equation}
To account for state generation error $\alpha = \frac 1 2 \norm{\rho - \sigma}_1$, 
we use the continuity of mutual information.
The inequalities of~\cite{Alicki2003,Audenaert2006} imply that
\begin{equation}
    \abs{
        I_\rho(A\!:\!B) - I_\sigma(A\!:\!B)
    } \le 10 \,\alpha \,\log(\min(d_A,d_B)/\alpha ) \quad \text{ assuming } \quad \alpha < \tfrac 1 2 \label{eq:continuityMI}
\end{equation}
for any bipartite states~$\rho$ and~$\sigma$ with Hilbert space dimensions~$d_A$ and~$d_B$.
If $A$ and $B$ together comprise an interval of length~$\ell$,
then we have $\min(d_A, d_B) \le q^{\ell/2}$.
Therefore,
\begin{equation}
    I_\sigma(A\!:\!B) \le 4m \log q + 5 \alpha \ell \log(q / \alpha ) \, . \label{eq:MIforGateCount}
\end{equation}

\subsection{Spatial profile}

Let $\ell$ be the length of an interval~$J$ in a circle of circumference~$L$ where
\begin{equation}
    \ell < \frac L 2 ( 1 - 10 \underbrace{\log_q 2}_{\delta} ) \, .
\end{equation} 
Here, we are concerned with states~$\sigma$ at time~$T < \ell /2$ 
starting with a product state
under the random circuit.
This allows us to use the infinite chain results for the purity in~\cref{thm:InfinitePurityBound}.
The condition~$T < \ell /2$ implies that the reduced density matrix on~$J$
is unitarily equivalent to the product of two states, 
one on each of the two boundary components of~$J$.
Hence, the ensemble-average purity~$P(T;\ell)$ of the semi-infinite chain
gives a bound on the second Renyi entropy~$S_2$ 
and in turn on the von Neumann entropy~$S$:
\begin{align}
  2T \log q \ge \overline{S(\ell)} \ge \overline{S_2(\ell)} 
  \ge - \log P(T;\ell) 
  = 2 T \log \frac{q^2 +1}{2q} > 2T ( 1 - \delta ) \log q \, ,
\end{align}
where the third inequality is by the convexity of~$\RR_{>0} \ni z \mapsto -\log z \in \RR$.
In addition, if we consider a subinterval of length~$x$ contained in~$J$,
then $x \le \ell$ and we can still use \cref{thm:InfinitePurityBound}:
\begin{align}
  x \log q \ge \overline{S(x)} \ge \overline{S_2(x)} 
  &\ge - \log \left[ \frac 1 {q^{x}} + \left(\frac{2q}{q^2 + 1} \right)^{2T} \right] 
  > x (\log q) - q^{x - 2T(1-\delta)} \, .
\end{align}
Therefore, as $q \to \infty$, we know $\delta \to 0$ 
and the von Neumann entropy
converges as
\begin{equation}
  \frac{S(x)}{\log q} \xrightarrow{\quad} \min(2T,x) \, .
\end{equation}

Divide the interval~$J$ into two subintervals $A$ of length~$x$ 
and $B$ of length $\ell - x$, where $x$ is even.
The mutual information $I_\sigma(A\!:\!B) = S_{\sigma}(A) + S_{\sigma}(B) - S_{\sigma}(AB)$
with log-base~$q$ converges to
\begin{equation}
  I(x;T,\ell) = \min(2T,x) + \min(2T,\ell-x) - \min(2T,\ell) 
\end{equation}
for each value of $x,T,\ell$ in the limit $q \to \infty$.
This quantity~$I$ as a function of~$x \in [0,\ell]$ has a graph of a trapezoid.
See \cref{fig:trapezoid}.
The top of the trapezoid moves up and down as a function of~$T$
while the sides of the trapezoid remain on the enveloping triangle.
The tip of the triangle has height~$\ell/2$ and is reached at $T = \ell / 4$ and $x = \ell / 2$.
This is when the trapezoid becomes a triangle.
The area under this graph is
\begin{equation}
  \calA 
  = \frac{\ell^2}{4} - \left(2T - \frac \ell 2 \right)^2
  = 2 T (\ell - 2 T) \, .
\end{equation}

\begin{figure}
    \centering
    \includegraphics[width=0.65\textwidth, trim={0ex 45ex 80ex 0ex}, clip]{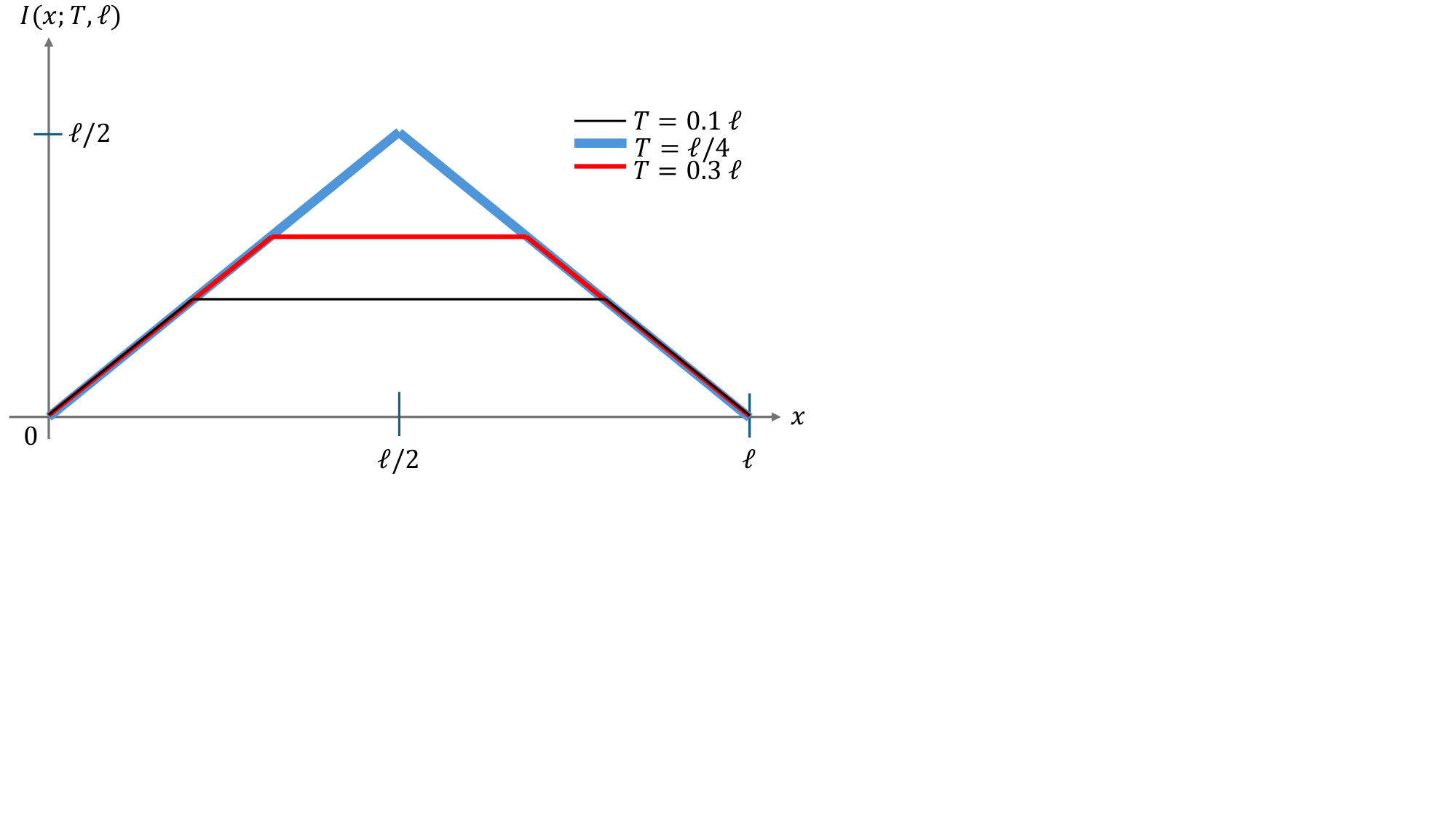}
    \caption{The graph of the mutual information profile $I(x;T,\ell)$,
    a trapezoid if~$T\neq \ell / 4$ and the full triangle if~$T = \ell/4$.
    Note that $I(x;T,\ell) = I(x;\frac 1 2 \ell - T, \ell)$.
    The area under the graph is a lower bound on the complexity.}
    \label{fig:trapezoid}
\end{figure}

\subsection{A complexity lower bound}

We now combine \cref{eq:MIforGateCount} with $I(x;T,\ell)$.
Let $\rho$ be the best approximation to~$\sigma$ 
using $\mathfrak C(\sigma)$ local quantum channels.
Along the cut at~$x$ within the interval~$J$,
we must have $m(x)$ nonidentity local quantum channels that straddle the cut,
where
\begin{equation}
    I(x;T,\ell) \le 4m(x) + 5 \alpha \ell 
\end{equation}
in the limit of large~$q$.
Summing over all even~$x$, we have
\begin{equation}
    (\ell - 2T) T \le  5 \alpha \ell^2 + 4 \sum_x m(x) \le 5 \alpha \ell^2 + 4 \mathfrak C(\sigma) \, .
\end{equation}
Here, the constant coefficients $4,5$ are unimportant and loose.
In the regime of interest where $T \propto \ell < \frac L 2 (1 - 10 \delta)$,
we conclude that the complexity~$\mathfrak C(\sigma)$ 
grows linearly with time~$T$ at least up to $T = \ell / 4$.

\section{Complexity growth according to holographic proposals}\label{secholographic}
In this section we will discuss proposed formulas for the complexity in holographic analogs of the random circuit setup. In particular, we consider a holographic 1+1 dimensional field theory on a circle of size $L$ that starts out in a high energy state with low initial entanglement: the B-state \cite{Calabrese:2005in}, which is $|\psi_0\rangle = e^{-\frac{\beta}{4}H}|B\rangle$ where $|B\rangle$ is a local product state. We take $\beta\ll L$ to be small compared to the size of the spatial circle.

The bulk description of this state is a one-sided black hole geometry in which the spacetime geometry ends on an end-of-the-world (eow) brane \cite{Hartman:2013qma}
as sketched in~\Cref{fig:bh}.
\begin{figure}[h]
    \hspace{3cm}
    \includegraphics[scale = .6, valign = c]{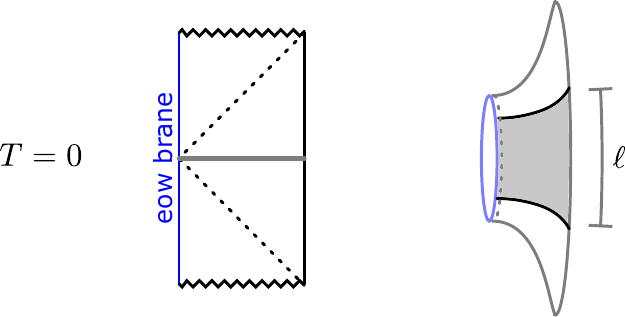}

    \hspace{3cm}
    \includegraphics[scale = .6, valign = c]{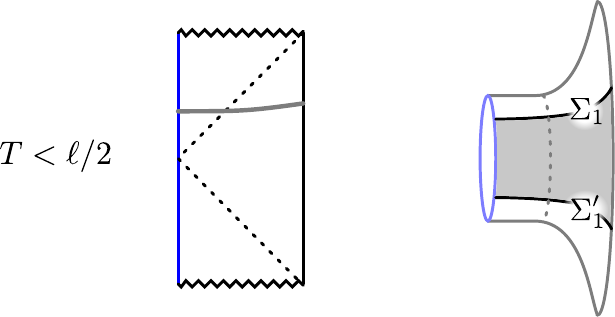}

    \hspace{3cm}
    \includegraphics[scale = .6, valign = c]{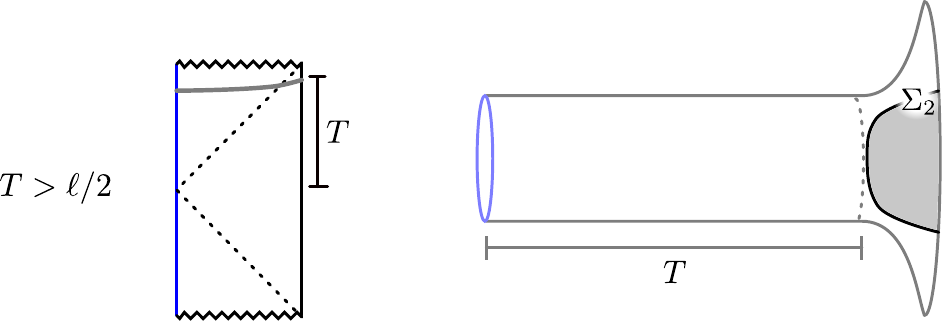}
    \caption{
        Spacetime geometry of a one-sided black hole.
        The figures on the right are sketches of the intrinsic geometry of the gray spatial slices
        at the times shown on the spacetime diagrams on the left.
        The spatial slice grows in time~$T$.
    }
    \label{fig:bh}
\end{figure}

In~\Cref{fig:bh},
we also sketched the HRT surface \cite{Hubeny:2007xt} that computes the entanglement entropy of the boundary interval of length $\ell$.\footnote{Here, we are stretching the truth a bit in the interest of clarity by drawing the surfaces as though they all live on the same spatial slice.} This surface is defined as the minimal extremal surface homologous to the boundary subregion, with the rule that the eow brane is considered homologically trivial. As a function of time, there is a transition between two possible choices of HRT surface \cite{Hartman:2013qma}. Their areas (lengths in this case) are
\begin{align}
\text{area}(\Sigma_1\cup \Sigma_1') &= 2\log\cosh\frac{\pi T}{\beta} \approx \frac{2\pi T}{\beta}\\
\text{area}(\Sigma_2) &= \log\sinh\frac{\pi \ell}{\beta} \approx \frac{\pi \ell}{\beta}.
\end{align}
where $\beta$ is the inverse temperature that the system thermalizes to at late time. Roughly, for $T < \ell/2$ the dominant surface is $\Sigma_1\cup\Sigma_1'$, with each of the two components extending from the boundary to the eow brane. For $T > \ell/2$ the minimizing surface is $\Sigma_2$, which connects the endpoints of the boundary interval while remaining outside the horizon of the black hole. As a function of time, the corresponding entropy grows linearly until reaching the thermal value at time $\ell/2$:
\begin{equation}
\text{entropy}(\ell) = s\times \begin{cases}2T & T < \ell/2 \\ \ell & T > \ell/2\end{cases}
\end{equation}
where $s \propto N_{\text{dof}}$ is the entropy density in the thermal state.

What does holography predict for the quantum complexity of the boundary subregion as a function of time? Complexity was introduced in the context of holography in \cite{Susskind:2014rva,Stanford:2014jda,Brown:2015bva,Brown:2015lvg}. Generalizations to the complexity of subsystems were proposed in \cite{Alishahiha:2015rta,Carmi:2016wjl,Agon:2018zso}, and applied to the time-evolved quantum quench in \cite{Chen:2018mcc}. Roughly, different proposals for holographic complexity amount to different reasonable measures of the ``size'' of the entanglement wedge. Here, the entanglement wedge is a spacetime region between the HRT surface and the asymptotic boundary. Its intersection with the spatial slice is shown by the shaded regions in \Cref{fig:bh}, and for the sake of an easy illustration, we will define the size of the entanglement wedge as the volume of this shaded region -- this corresponds roughly to the ``CV'' prescription \cite{Alishahiha:2015rta}. If we further define the complexity by subtracting the value in the thermal state, then the complexity is approximately just the volume of the gray shaded region that lies behind the horizon (dashed lines in \Cref{fig:bh}). Explicitly, for $\ell < L/2$
\begin{align}\label{compsmall}
\mathfrak{C}(\ell) &= \frac{s}{\beta}\times\begin{cases} \ell T & T < \ell/2 \\ 0 & T > \ell/2\end{cases}.
\end{align}
For the complementary region of size $L -\ell > L/2$, the entanglement wedge is the complement of the gray region within the spatial slice, and the corresponding formula is
\begin{align}
\mathfrak{C}(L-\ell) &= \frac{s}{\beta}\times \begin{cases} (L-\ell) T & T < \ell/2 \\ L T & T > \ell/2\end{cases}\hspace{20pt}.
\end{align}
It is important to note that the holographic complexity conjectures (especially for subsystems) have not been derived, so these formulas are quite tentative.

The most dramatic feature of the above formulas is the sudden collapse in the complexity of the smaller subsystem at time $T = \ell/2$. This sudden change is due to the following: when the entanglement surface jumps from $\Sigma_1\cup\Sigma_1'$ to $\Sigma_2$, the entanglement changes continuously, but the entanglement wedge itself changes abruptly.\footnote{Whether the entanglement wedge evolves continuously or discontinuously depends on details, and in the AdS$_3$ quench from the vacuum \cite{Chen:2018mcc} the transition is continuous. See earlier \cite{Albash:2010mv,Balasubramanian:2010ce,Liu:2013iza}.} To see what this implies for complexity, consider a family of slightly different time evolution problems, where we modify the CFT Hamiltonian by a time-dependent sequence of operator insertions. We will assume these operators are light enough that they don't change the geometry; they just add particles that fall into the black hole at different times
as depicted in~\Cref{fig:drawing2}.
\begin{figure}[h]
    \centering
    \includegraphics[scale = .75, valign = c]{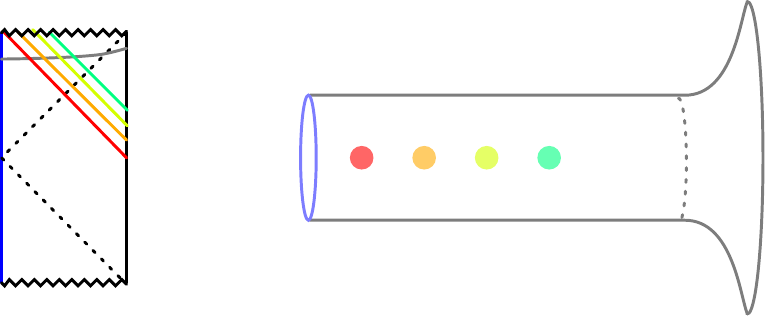}
    \caption{Light operator insertions at different times.}
    \label{fig:drawing2}
\end{figure}
Entanglement wedge reconstruction implies that the density matrix of the boundary subregion can distinguish evolution histories that lead to different sets of particles within the entanglement wedge \cite{Faulkner:2013ana,Jafferis:2015del}. For $T < \ell/2$, the entanglement wedge of the subregion has volume $\ell T$, so we can make $n^{\text{const}.\ell T}$ different states, where $n$ is the number of light operators one can choose from. This gives a lower bound on the complexity of order $\propto\log(n) \ell T$. This can be numerically large, but it is small compared to (\ref{compsmall}), which scales as $N_{\text{dof}}\ell t$, so this argument is not strong enough to validate (\ref{compsmall}).

A second interesting feature of these formulas is the fact that for early times, the complexity of either system grows slower than $L T$ -- this suggests that an improvement exists over the naive preparation recipe of creating the global pure state and then tracing out. We will return to this at the end of the next section.

\section{Can a subsystem remember the circuit?}\label{secreplica}
In the previous section, we offered intuition for the sudden decrease in complexity at time $T = \ell/2$ based on a similarly sudden decrease in the subsystem's memory of its own evolution history. Here, we will use the replica trick to analyze this for the random circuit, computing
\begin{align}
    F(\rho,\sigma) \equiv \tr\left(\sqrt{\sqrt{\rho}\sigma\sqrt{\rho}}\right)= \text{continuation}_{n\to 1/2} \tr\big[(\rho\sigma)^n\big]
\end{align}
where $\rho,\sigma$ are the density matrices for the subsystem of size $\ell$ produced by slightly different random circuits. If the fidelity is small, the subsystem can remember the difference.

\subsection{Warmup problem}
Before getting to the brickwork circuit, let's study a warmup problem in a Hilbert space of dimension $qQ$ with
\begin{align}
    \rho &= \tr_{q}\big[U |0\rangle\langle 0|U^\dagger\big]\\
    \hspace{30pt} 
    \sigma &= \tr_{q}\big[V U |0\rangle\langle 0|U^\dagger V^\dagger\big]
\end{align}
where $V$ is some fixed unitary and $U$ is Haar random.

The integral over $U$ of $\tr[(\rho\sigma)^n]$ simplifies if $q,Q\gg 1$. The leading terms in the Weingarten formula reduce to ``Wick contractions'' of $U$s with $U^\dagger$s. In other words, we can forget unitarity and replace the $U$ matrices by random tensors, reducing the problem to the one studied in \cite{Hayden:2016cfa}. These Wick contractions can be labeled by permutations $\pi$ of $2n$ objects, with the idea that the $i$-th copy of $U$ is contracted with the $\pi(i)$-th copy of $U^\dagger$. To analyze the sum over permutations, let's temporarily set $V = \one$. Then each permutation would contribute
\begin{equation}
    \frac{1}{(Qq)^{2n}}q^{\text{\# cycles in $\pi$}}\cdot Q^{\text{\# cycles in $\pi\circ \tau^{-1}$}}
\end{equation}
where $\tau$ is the cyclic permutation. For $q \gg Q$, this is dominated by the identity permutation, and this particular contribution is actually independent of $V$. On the other hand, for $Q \gg q$, the dominant contribution is from $\pi = \tau$, and that contribution depends strongly on $V$. Including these two extreme cases explicitly, one has
\begin{equation}
    \int \mathrm{d} U \tr_Q\big[(\rho\sigma)^n\big] \approx  \frac{1}{(Qq)^{2n}}\left\{\tr_q\left[\left(\tr_Q(V)\tr_Q(V^\dagger)\right)^n\right] + \dots +  Q q^{2n}\right\}.\label{terms}
\end{equation}
Analytically continuing the leading term to $n = 1/2$ gives 
\begin{equation}
    \int \mathrm{d} U F(\rho,\sigma) \approx \begin{cases} \frac{\tr_q\sqrt{\tr_Q(V)\tr_Q(V^\dagger)}}{qQ} & q \ll Q \\
    1 & q \gg Q\end{cases}.
\end{equation}
Corrections to this formula are suppressed by powers of $q,Q$ and/or $Q/q$, see Appendix E of \cite{Penington:2019kki} for some details. Note that if $V$ is such that the term on the first line is small enough, such corrections could compete, leading to a different answer than the one shown, although still small.

\subsection{$q,Q$ brickwork circuit}
For the random brickwork circuit, the replica problem is complicated by the fact that there are multiple patterns of Wick contractions that contribute at leading order. In the computation of the Renyi 2-entropy, this manifests in the sum over trajectories for the entanglement cut connecting the endpoints of the interval to the initial conditions.

One can avoid this by modifying the random circuit to remove the degeneracy. Specifically, choose even sites to have dimension $q$, and odd sites to have dimension $Q$. See~\Cref{figcircuitredlines}.
Then take
\begin{equation}
    q,Q\to \infty, \hspace{20pt} Q/q \gg 1.
\end{equation}

\begin{figure}
    \centering
    \includegraphics[scale = .8, valign = c]{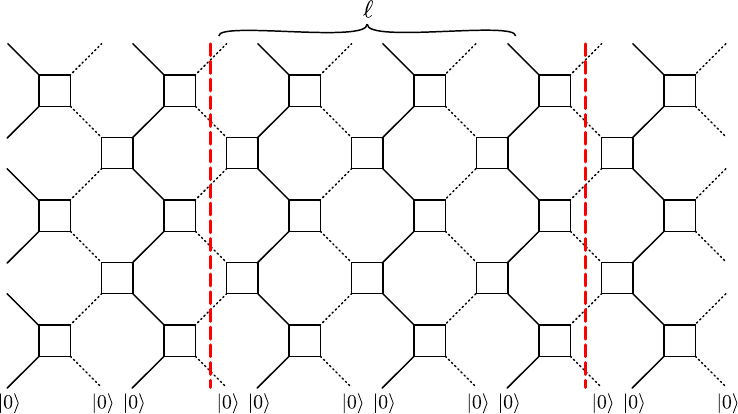}
    \caption{Random circuit on qudits of alternating dimensions $q$ and $Q$ where $q \ll Q$. 
    The solid lines mean Hilbert spaces of dimension $Q$, and the dotted lines mean Hilbert spaces of dimension $q$.
    The red dashed lines indicate the entanglement cut through this circuit that dominates for short times. Note that because $Q \gg q$, only $q$ lines are cut.}    
    \label{figcircuitredlines}
\end{figure}

Let's first study the Renyi entropies for the subsystem of $\ell$ sites. For each unitary gate in the quantum circuit, we have a permutation $\pi$.
For large $q,Q$, the dominant configuration is one where the permutations are either identity or cyclic, with a domain wall separating the two regions, \cite{Hayden:2016cfa}. There are two possible optimal configurations of the domain wall. One possibility is that the domain wall cuts horizontally without dipping down into the tensor network at all.\footnote{There are actually a family of degenerate cuts that dip down into the circuit but remain in the past domain of dependence of the $\ell$ interval. The sum over these would appear to give an enhancement to the Renyi entropy. However, these extra contributions are canceled by negative contributions from subleading non-Wick-contraction terms in the Weingarten formula. (These are the terms that distinguish truly unitary tensors from ordinary random tensors.) One can avoid this complication by using unitarity to remove the tensors in the past domain of dependence of the interval, making the cut unique.} The other possibility corresponds to the dashed red lines in \Cref{figcircuitredlines}. Depending on $\ell,T$, one or the other of these will dominate, and
\begin{equation}
    \tr \rho_{\ell}^{n}  = \text{max}\left\{\frac{1}{(qQ)^{(n-1)\ell/2}}, \frac{1}{q^{2(n-1)T}}\right\}.\label{secondtermdominates}
\end{equation}

Finally, we are ready to study
\begin{equation}
    \tr \left[(\rho_\ell\sigma_\ell)^n\right].
\end{equation}
Here, the only difference $\rho,\sigma$ is that one brick in the circuit has been modified $U \to V U$. In the late-time phase where $q^{2T} \gg (qQ)^{\ell/2}$, so that the first term in (\ref{secondtermdominates}) dominates, these insertions of $V$ operators have no effect because the permutation field was $\pi = 1$ everywhere on the lattice and $V$ trivially cancels out. Continuing to $n = 1/2$, one finds $F(\rho_\ell,\sigma_\ell) = 1$. In this regime, the density matrix is maximally mixed and has no memory of the specific circuit that prepared it.

In the early-time phase, the dominant permutation is modified to the cyclic permutation $\pi = \tau$ inside the entanglement wedge (the region between the red lines of \Cref{figcircuitredlines}). Now the answer depends on whether the $V$ modification takes place inside or outside the entanglement wedge\footnote{If the perturbation is inside but adjacent to the cut then one finds a slightly more complicated formula similar to the first term in (\ref{terms}).}
\begin{equation}
    \tr[(\rho_\ell\sigma_\ell)^n]  = \frac{1}{(qQ)^{(2n-1)\ell/2}}\begin{cases} 1 & V \text{ outside} \\ \left(\frac{|\tr V|}{qQ}\right)^{2n} & V \text{ inside}
    \end{cases}.
\end{equation}
Taking the limit $n \to 1/2$ and combining with the late time case, we get
\begin{equation}
    F(\rho_\ell,\sigma_\ell) = \begin{cases} 1 & q^{2T}\gg (qQ)^{\ell/2} \ \text{ or } \ V \ \text{ outside} \\
    \frac{|\tr V|}{qQ} & q^{2T} \ll (qQ)^{\ell/2} \ \text{ and } \ V \ \text{ inside}
    \end{cases}.
\end{equation}
So in the early time phase, circuits that differ within the entanglement wedge lead to approximately orthogonal density matrices. In other words, the density matrix remembers the portion of its preparation circuit that was inside the entanglement wedge, but it doesn't remember the part that was outside.

This gives some intuition for the sudden collapse of complexity. In order for the density matrix to be able to distinguish circuits that differ in any of $\ell T$ places, exponentially many different density matrices must be possible, leading to a $\ell T$ counting lower bound on the complexity. This distinguishing power disappears suddenly at time $\ell/2$. To turn this into a rigorous bound,
we would need enough control over the probability distribution for the fidelity
to show something like $\Pr[F(\rho_\ell, \sigma_\ell) > 0.5] \le (qQ)^{-\text{(const.)} \ell T}$.
In the next section we explore a similar question in a simpler setting.

\begin{figure}[h!]
    \centering
    \includegraphics[scale = 0.6, valign = c]{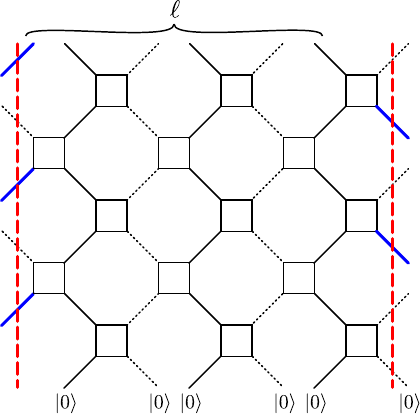}
    \caption{If the density matrix on the interval of length~$\ell$ does not remember gates outside, we may average over these gates,
    effectively replacing the thick blue lines by maximally mixed states.}
    \label{fig:drawing6}
\end{figure}

This result also sheds some light on another puzzling feature of the holographic formulas: before $T = \ell/2$, the complexity of either interval is proportional to its length times $T$, rather than the total length $L$. Because the density matrix does not remember the gates outside the entanglement wedge, we can average over these gates. This is equivalent to replacing the lines coming into the entanglement wedge 
(the thick blue lines in \Cref{fig:drawing6})
by maximally mixed states.
This gives a recipe with complexity $\ell T$. Note that this suggests a sudden jump in the complexity of the larger interval at time $\ell/2$. This is because after that time, the entanglement wedge of the larger interval expands to the entire circuit (except for the past domain of dependence of the small interval), so the simplified recipe stops being possible.

Here we give an argument that at late time, the complexity growth rate of the larger system (call it $B$) should indeed be proportional to $L$ and not $(L-\ell)$.
Once the smaller subsystem~$A$ thermalizes, 
the reduced density matrix on~$A$ is constant in time.
Uhlmann's theorem implies that all subsequent changes 
by unitary evolution~$U_{AB}$ on the full system~$AB$
can be implemented by some unitary~$V_B$ acting only on~$B$.
In the random circuit model 
in which the thermal state on~$A$ is the maximally mixed state,
the construction of~$V_B$ from~$U_{AB}$ can be seen as follows.
Let $\ket{\psi(T)_{AB}}$ be the state of~$AB$ at time~$T$.
The entanglement between $A$ and $B$ is maximal when $A$ is thermal,
and we have 
\begin{equation}
    \ket{\psi(T)_{AB}} = (\one_A \otimes V_B(T)) \ket{\Phi_{AB}} 
    = (\one_A \otimes V_B(T)) \sum_{a} \ket{a_A} \ket{a_{A'} 0_{B \setminus A'}}
\end{equation}
where $A'$ is a subregion of~$B$ that has the same Hilbert space dimension as~$A$.
The EPR state $\sum_a \ket{a_A}\ket{a_{A'}}$ is the $+1$-eigenstate of
all operators of form~$U_A \otimes U^T_{A'}$.
Hence,
\begin{equation}
    (O_A \otimes O_B) \ket{\psi(T)_{AB}} = (\one_A \otimes O_B V_{B} O^*_{A'}) \ket{\Phi_{AB}} \, .
\end{equation}
If $U_{AB}$ is the time-evolution unitary for small time~$\Delta$, 
we may write it as $U_{AB} = W_{\partial A} (O_A \otimes O_B)$
for some unitary~$W_{\partial A}$ acting on the interface between~$A$ and~$B$.
We see that the effect of~$U_{AB}$ is almost reproduced (except for~$W_{\partial A}$)
by~$V_B(T+\Delta) = O_B V_{B}(T) O^*_{A'}$.
Assuming absolute incompressibility of the unitaries,
the complexity change from $V_{B}(T)$ to $V_B(T + \Delta)$ 
is the complexity of~$U_{AB}$ which is $L$ rather than~$L - \ell$.

\section{Counting distinguishable states}\label{counting}
Here we consider various finite sets of states
that are pairwise as distinct as possible 
under the constraint that every state in 
the set has maximal entropy minus a few bits.
This shows that there is, in principle, enough room for large complexity
for a smaller subsystem just before it thermalizes.

\subsection{Using unitary $k$-designs}

Consider an ensemble of $d$-dimensional mixed states (with $d$ even):
\begin{align}
    \tau_U &= \frac 1 d 
            U 
            \diag( 
                \underbrace{1 + \alpha, \ldots, 1+\alpha}_{d/2} ,\, 
                \underbrace{ 1 - \alpha, \ldots, 1 - \alpha}_{d/2}
            ) 
            U^\dagger \, ,\\
    \alpha &= \norm*{\tau_U - \frac \one d }_1 \, .\nonumber
\end{align}
where $U \in \U(d)$ and $0 \le \alpha \le 1$.
The ensemble $\{\tau_U\}_U$ is defined by a probability distribution~$\nu$ of~$U$,
and we suppose that $\nu$ is an $\epsilon$-approximate unitary $k$-design~$U \sim \nu$ in the sense of~\cref{eq:approxDesignEps}. 
We ask how many states~$\tau_{U_i}$ we can find in this ensemble 
such that $\norm{\tau_{U_i} - \tau_{U_j}}_1 \ge \beta$ whenever~$i \neq j$.
Such a finite subset may contain more than one element only if~$\beta \le 2\alpha$.

With a rank-$\frac d 2$ projector $P = \diag(1,1,\ldots,1, 0, 0, \ldots, 0)$,
it is easy to see~\cite[Eq.(31)]{Haah2015} that 
\begin{equation}
    \norm{\tau_U - \tau_\one}_1 \ge 2 \alpha( 1 - 2 d^{-1} \Tr(PUPU^\dagger) ) .
\end{equation}
Then, 
\begin{align}
    \Pr_{U \sim \nu} \left[ ~
        \norm{\tau_U - \tau_\one}_1 < \beta
    ~\right]
    &\le
    \Pr_{U \sim \nu} \left[~
        \frac{4}{d} \Tr(PUPU^\dagger) > 2 - \frac \beta \alpha
    ~\right]\\
    &\le \frac{1}{( 2 - \beta \alpha^{-1} )^{k}} \avg_{U \sim \nu} \left( \frac 4 d \Tr(PUPU^\dagger) \right)^k & (\text{Markov}) \nonumber\\
    &\le 
    \frac{1+\epsilon}
    {( 2 - \beta \alpha^{-1} )^{k}} \avg_{V \sim \U(d)} \left( \frac 4 d \Tr(PVPV^\dagger) \right)^k &
    \text{by~\cref{eq:approxDesignEps}} \nonumber\\
    &\le
    \frac{1+\epsilon}{( 2 - \beta \alpha^{-1} )^{k}}
    \min(e^{2 k^2 / d^2 },\, k!) & \text{by~\cref{thm:randomProjectorOverlap}} \, . \nonumber
\end{align}

In the regime where $k < d$ and $\epsilon < 1$,
the numerator is bounded above by~$2 e^2 < 15$.
We conclude that in this regime 
there exist $N$ states in the $\alpha$-ball around~$\one_d / d$
that are pairwise $\beta$-separated in $1$-norm where
\begin{equation}
    N > \frac 1 {15} \left(2 - \frac \beta \alpha \right)^k.
\end{equation}

\subsection{Using full unitary group}

\begin{remark}
If we drop the condition that $U \sim \nu$ is a unitary $k$-design,
and consider the full set of~$\tau_U$ with arbitrary~$U$~\cite[Packing net II]{Haah2015},
then the maximum number~$N_0$ of pairwise $\beta$-separated states~$\tau_{U_i}$
is lower bounded by~\cref{thm:randomPQOverlapGaussian} as, when $\beta < \alpha$,
\begin{equation}
    N_0 \ge \exp\left( \frac{(1-\beta \alpha^{-1})^2}{16} d^2 \right)\, .
\end{equation}
\end{remark}

\begin{remark}
    Let $P$ be a projector of rank~$r$ and dimension~$d$.
    It is not difficult to show that $\norm{UPU^\dagger - P}_1 \ge 2 \Tr(UPU^\dagger (\one - P))$
    for any~$U \in \U(d)$.
    Using~\cref{thm:randomPQOverlapGaussian}, we find that if $r < \epsilon d$,
    \begin{equation}
        \Pr_{U \sim \U(d)} \left[ \norm*{U \frac P r U^\dagger - \frac P r}_1 < 2  - 2 \epsilon \right]
        \le
        \exp\left( - \frac{r d}{2} \frac{ (\epsilon - \frac r d)^2}{1 - \frac r d} \right) \, .
    \end{equation}
    This implies that there are at least $\exp( \frac 1 2 (\epsilon - \frac r d)^2 rd )$
    pairwise almost orthogonal states of rank~$r$.
\end{remark}

\begin{remark}
    Continuing in the same setting, let $F$ be the fidelity between two states obtained from random projectors of rank $r$ with $r < d/2$, and suppose we require $F<\varepsilon$. Then for $r\gg 1$ one can use random matrix methods to show (see \cref{app:rmt})
    \begin{align}\label{var}
         \log \Pr(F) &= -r^2\cdot(\text{some function of $F$, $r/d$})\\
         &=-r^2\cdot\left[\frac{\pi^2}{4\frac{r}{d}(1-\frac{r}{d})}\big(F-\mathbb{E}(F)\big)^2  + \dots\right]
    \end{align}
    where the dots are cubic and higher in $F - \mathbb{E}(F)$. The expectation value is
    \begin{equation}
        \mathbb{E}(F) = \frac{2}{\pi}\left(\sqrt{\frac{1-w}{w}} - \frac{1-2w}{w}\text{arcsin}(\sqrt{w})\right), \hspace{20pt} w \equiv \frac{r}{d}.
    \end{equation}
To make the fidelity between a typical pair small, we need $r \ll d$. Suppose the tolerance parameter $\varepsilon$ is much larger than the typical fidelity, $\sqrt{r/d} \ll \varepsilon \ll 1$. Then (\ref{var}) gives $\log \Pr(\varepsilon) = -\frac{\pi^2rd}{4}\varepsilon^2$ so there are $\exp\left(\frac{\pi^2}{4}\varepsilon^2rd \right)$ states with all pairwise fidelities less than $\varepsilon$.
\end{remark}

Note the following slightly counterintuitive aspect of this result. 
A quantum channel that outputs a maximally mixed state for all inputs has zero capacity. 
Moreover, if the output has maximal entropy minus a few bits, 
then the quantum capacity is at most a few bits. 
In this sense, such density matrices contain little quantum information.
However, the results above imply that one can find a very large number of 
nearly orthogonal density matrices with entropy just a few bits less than maximal.
In other words, channels with such outputs can have a high capacity for the classical and quantum identification problems, accurately (but not exactly) encoding the geometry of high dimensional 
subspaces~\cite{winter2005quantumclassicalmessageidentification,Hayden_2012,Hayden:2017xed}.

\section*{Acknowledgements}
We are grateful to Henry Lin and Patrick Hayden for discussions. We thank Yale Fan, Nicholas Hunter-Jones, Andreas Karch, and Shivan Mittal for coordinating submission of their paper \cite{Fan2025}. This work was supported in part by 
DOE grants DE-SC0021085, DE-SC0026143, and DE-SC0025934,
and by a grant from the Simons foundation (926198, DS).

\appendix 

\section{Deferred Proofs}

\subsection{Projector overlap}\label{app:randomProjLemma}

The following is our summary of the argument presented in~\cite{Hayden2003,Hayden2004}
for the proof of~\cite[Lemma~III.5]{Hayden2004},
which is reproduced below for readers' convenience.

\begin{proof}[Proof of~\cref{thm:randomProjectorOverlap}]
  For the first assertion, we may assume that $A,B$ are diagonal with~$1$ in the first~$a,b$ entries, respectively.
  Consider a vector of independent Gaussian entries, 
  $\ket{\tilde x} = (x_1 + \ii x_2, \ldots, x_{2ad-1}+\ii x_{2ad}) \in \CC^d \otimes \CC^a$.
  If we trace out $\CC^a$ from~$\proj{\tilde x}$ we obtain an unnormalized density 
  operator~$\tilde \rho \in \CC^{d \times d}$,
  whose trace has expectation value~$ad$.
  Since the distribution of~$\ket{\tilde x}$ is invariant under unitary multiplication,
  the distribution of~$\tilde \rho$ is invariant under unitary conjugation.
  The eigenspectrum of~$\tilde \rho$ has at most $a$ nonzero values since it is obtained 
  by tracing out the $a$-dimensional tensor factor.
  Hence, the distribution of $\tilde \rho$ is reproduced by 
  the distribution of~$U \tilde \eta U^\dagger$
  where $\tilde \eta$ follows some distribution of a diagonal matrix,
  of which all but the first $a$ entries are zero,
  and a Haar random~$U \in \U(d)$.
  Since the Haar distribution is invariant under multiplication by any permutation matrix,
  we may assume that the distribution of~$\tilde \eta$
  is invariant under permuting the first $a$ diagonal entries.
  This implies that $\avg_{\tilde \eta} \tilde \eta = d A$.
  Now,
  \begin{align}
    \avg_{x} \phi\left(\sum_{i=1}^{2ab} x_i^2 \right)
    &= 
    \avg_{x} \phi(\bra{\tilde x} B \otimes I_a \ket{\tilde x}) 
    = \avg_{\tilde \rho} \phi \Tr(B \tilde \rho) 
    = \avg_{U,\tilde \eta} \phi \Tr(B U \tilde \eta U^\dagger) \\
    &\ge \avg_U \phi \avg_{\tilde \eta} \Tr(B U \tilde\eta U^\dagger)
    = \avg_U \phi \Tr(d B U A U^\dagger) \nonumber
  \end{align}
  where the inequality is because $\phi$ is convex.
  The function $x \mapsto \phi(x / ab)$ is convex if and only if $\phi$ is.
  
  The second assertion is straightforward calculation with a probability density function $\frac{e^{-x^2}}{\sqrt \pi} $:
  \begin{align}
    \avg_{x_i} \left(\sum_{i=1}^{2m} x_i^2 \right)^k
    &=
    \int_{\RR^{2m}} \frac{\rd x_1}{\sqrt \pi} \cdots \frac{\rd x_{2m}}{\sqrt \pi} 
    (-\partial_{\lambda})^k |_{\lambda=1} \exp\left(-\lambda\sum_{i=1}^{2m} x_i^2\right) \nonumber\\
    &=
    (-\partial_{\lambda})^k |_{\lambda=1} 
    \left(
      \int_{\RR^{2}} \frac{\rd x_1 \rd x_2}{\pi} e^{-\lambda(x_1^2+x_2^2)}
    \right)^m
    = 
    (-\partial_{\lambda})^k |_{\lambda=1} 
    \lambda^{-m} \\
    &=
    m (m + 1) \cdots (m + k - 1) \, . \nonumber \qedhere
  \end{align}
\end{proof}

\begin{lemma}[Lemma~III.5 of~\cite{Hayden2004}]\label{thm:randomPQOverlapGaussian}
    Let $A$ and $B$ be orthogonal projectors on~$\CC^d$ of rank~$a$ and~$b$, respectively.
    Define $f(z) = z - \ln(1+z) = \frac 1 2 z^2 + \cdots$. For a Haar random~$U \in \U(d)$,
    \begin{align}
        \forall z > 0 &:~ 
            \Pr_U \left[ \frac d {ab} \Tr(U A U^\dagger B) \ge 1 + z \right] \le \exp(-ab f(z)) \, ,\\
        \forall z \in (-1,0) &:~
            \Pr_U \left[ \frac d {ab} \Tr(U A U^\dagger B) \le 1 + z \right] \le \exp(-ab f(z)) \, . \nonumber
    \end{align}
\end{lemma}
\begin{proof}
    Let $Z = \frac d {ab} \Tr(U A U^\dagger B)$.
    Markov's inequality implies that 
    $\Pr[Z \ge 1 + z] = \Pr[e^{\xi Z} \ge e^{\xi(1+z)}] \le \avg_U e^{\xi Z} e^{-\xi(1+z)}$
    and $\Pr[Z \le 1 + z] = \Pr[e^{-\xi Z} \ge e^{-\xi(1+z)}] \le \avg_U e^{-\xi Z} e^{\xi(1+z)}$
    for any~$\xi > 0$.
    Use $\phi : y \mapsto \exp(\pm \xi y)$ in~\cref{thm:randomProjectorOverlap} 
    with~$\xi = ab \abs z / (1+z)$,
    and calculate Gaussian integrals.
\end{proof}

\subsection{Purity of a finite interval on an infinite line}\label{app:Purity}

It was shown
in~\cite{Nahum2017} that the ensemble-averaged purity is
precisely the sum\footnote{Here $q$ is arbitrary integer, not necessarily a power of~$2$.} 
\begin{equation}
  \sum_{C} \left(\frac{q}{q^2 + 1}\right)^{\abs {\partial C}}
\end{equation}
where $C$ is a configuration of binary variables ($0$ or $1$), 
one for each of the random gates in the circuit and 
sites located at the top in a spacetime diagram,
subject to a boundary condition that 
all site binary variables outside the interval are $0$ (up) and all others $1$ (down)
and where the weight~$\abs{\partial C}$ is the length of domain wall,
calculated by the rule shown in~\cref{fig:domainwall}.

\begin{figure}[htb]
    \centering
    \includegraphics[width=0.8\textwidth, trim={0ex 15ex 5ex 19ex}, clip]{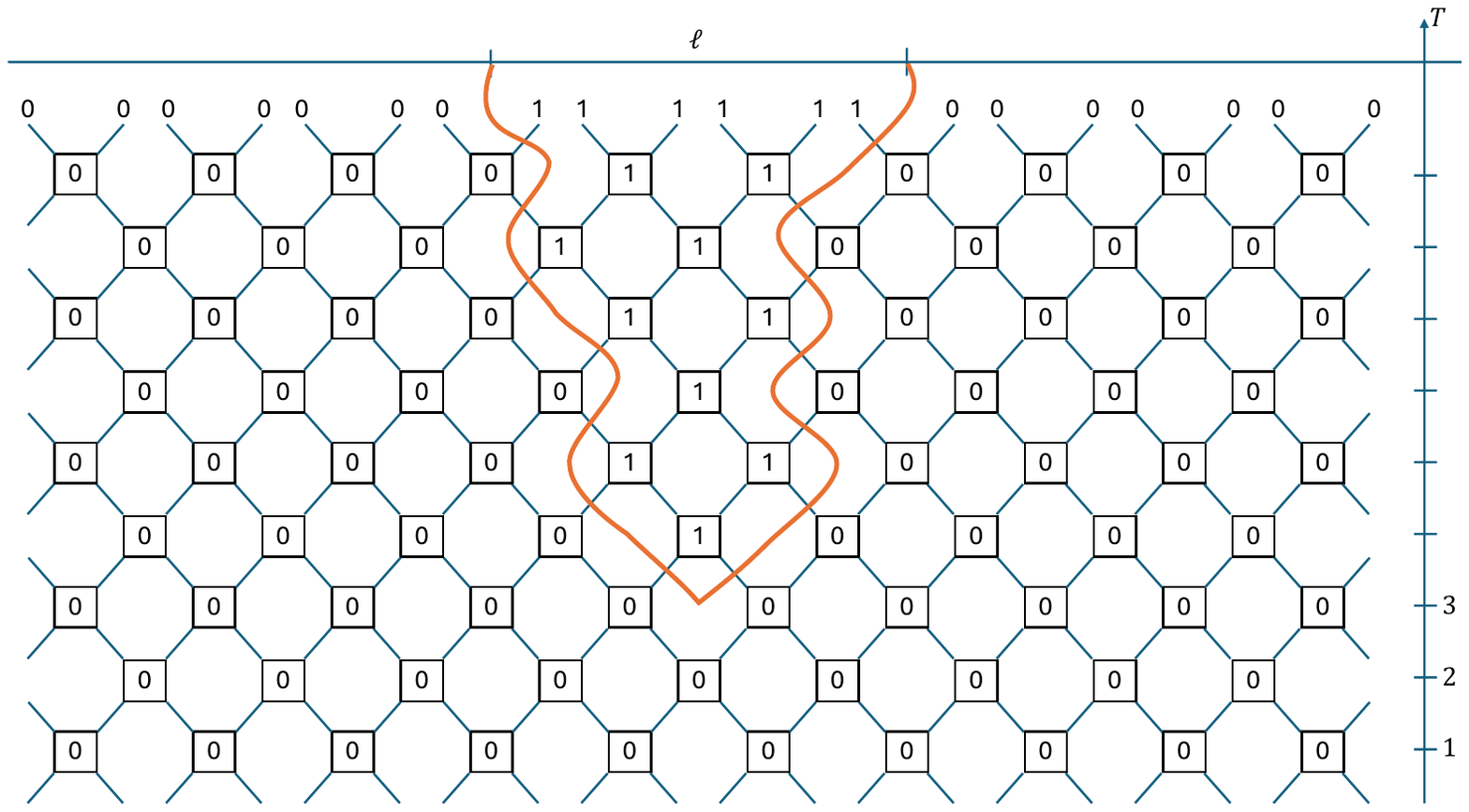}
    \caption{Two domains walls must propagate from the top to bottom and must not go up.
        They may or may not merge in the bulk.
        Each have weight~$7$ in the figure, giving total weight~$14$.}
    \label{fig:domainwall}
\end{figure}

\noindent
There are domain configurations with a single-component domain wall
or two-component domain wall, 
in the latter case of which the end points lie at the bottom.
We always think of domain walls propagating from the top to the bottom.
The weight of a configuration is determined by the time coordinate of the merger point.
If the merger happens at the slice of time coordinate $t > 0$,
the length of the domain wall is $2(T-t) < 2T$.
If two domain walls reach $t=0$, then they may or may not meet,
but in any case the total length of the domain walls is $2T$.

The evaluation of the purity boils down to the path counting,
starting from the boundary points of the interval,
propagating downward.
For concrete counting we introduce a coordinate system on a square lattice, 
as shown in~\cref{fig:domainwallcoord},
in which the boundary points of the interval are at $(\ell/2,0)$ and $(0,\ell/2)$.
The walkers go up or right only, and once they meet, the progression terminates.
The time coordinate~$t$ in the original spacetime of a point $(x,y)$ in the new coordinate system
is $t = T - x - y + \half \ell$. Define
\begin{equation}
  \eta = \frac{q}{q^2 +1} \, .
\end{equation}

\begin{figure}[hbt]        
        \centering
        \includegraphics[width=0.7\textwidth, trim={0ex 25ex 60ex 15ex}, clip]{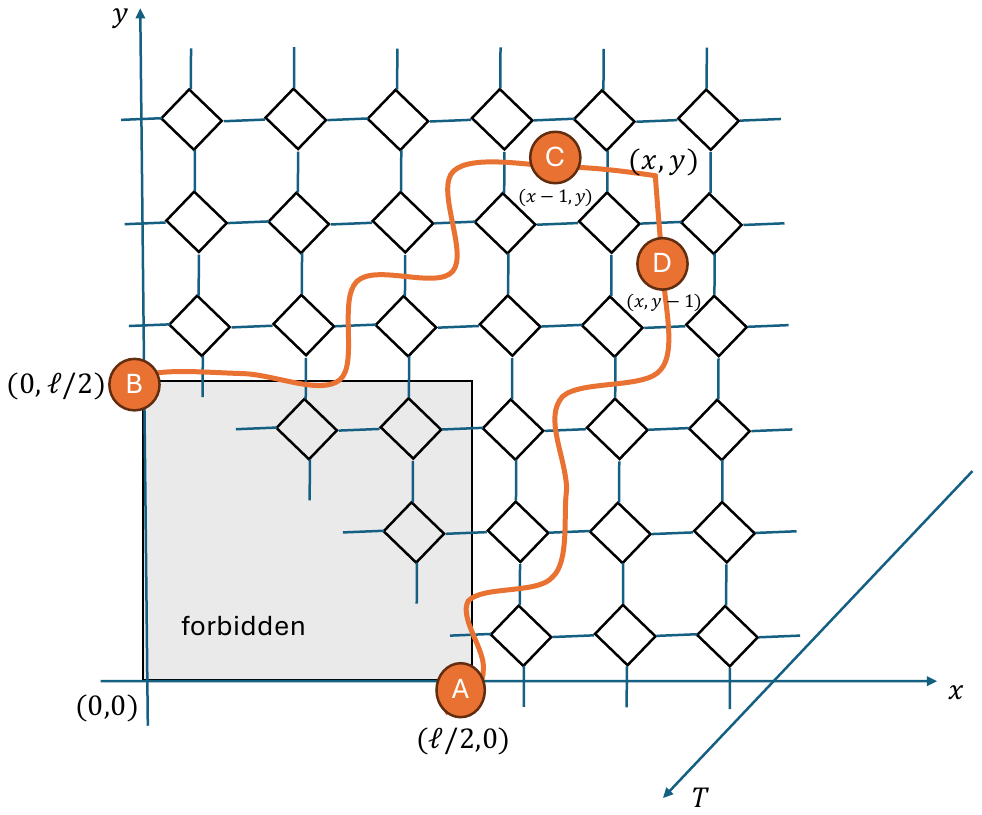}
  \caption{Coordinate system for counting the domain walls.
  The points $C$ and $D$ are present only if the two domain walls merge.
        The meeting point is neither $C$ nor $D$; 
        these points are immediately before the merger
        and are used to evaluate~$N(x,y;\ell)$.}
  \label{fig:domainwallcoord}
\end{figure}

The number of paths that meet at $(x,y)$ for the first time where $x,y \ge \half \ell > 0$
can be counted by reflection trick:
\begin{align}
  N(x,y;\ell) 
  &= 
  \#\{ (\half \ell,0) \to (x,y-1)\}\cdot \#\{ (0,\half \ell) \to (x-1,y) \} \nonumber\\
  & \qquad -
  \#\{ (\half \ell,0) \to (x-1,y)\}\cdot \#\{ (0,\half \ell) \to (x,y-1) \} \\
  &=
  \binom{\tau-1}{y-1} \binom{\tau-1}{x-1}
  -
  \binom{\tau-1}{y} \binom{\tau-1}{x}  \ge 0 & (\tau = x+y-\half \ell) \nonumber
\end{align}
The total length of any such merging paths is~$2\tau$.
We define
\begin{equation}
  N(x,y;0) = \begin{cases} 1 & \text{if } x = y = 0 \\ 0 & \text{otherwise} \end{cases} .
\end{equation}
This makes sense in view of the path counting.
Note that $N(x,y;\ell) = 0$ if either $x < \frac \ell 2$ or $y < \frac \ell 2$.

\begin{lemma}\label{lem:infiniteInfinitePurity}
  For any real $q \ge 1$ and integer~$\ell \ge 0$, we have
  \begin{equation}
    \sum_{x, y = \frac\ell 2}^{\infty} N(x,y;\ell) \, \eta^{2x +2y - \ell}  = \frac{1}{q^\ell} \, .
  \end{equation}
\end{lemma}

\begin{proof}
  Assume $q > 1$.
  Since $\eta = q / (q^2 + 1) < \half$ and $N(x,y;\ell) \le 2^{2\tau} \le 2^{2x+2y}$,
  the infinite sum converges absolutely. Consider
  \begin{equation}
    Q(\ell) = \sum_{a,b = 0}^\infty N(a + \tfrac \ell 2, b + \tfrac \ell 2; \ell) \eta^{2a + 2b}.
  \end{equation}
  We claim that $Q(0) = 1$, $Q(\ell)Q(2) = Q(\ell+2)$, and $Q(2) = (1+q^{-2})^2$,
  which will complete the proof.
  
  The first claim $Q(0) = 1$ is immediate by definition of~$N(x,y;\ell)$.

  The second claim with $\ell > 0$ follows from
  \begin{equation}
    \sum_{a+b+c+d = f} N(a+ \tfrac \ell 2, b+ \tfrac \ell 2; \ell) N(c + 1, d+1; 2)
    =
    \sum_{u+v = f} N(u + \tfrac \ell 2 + 1, v + \tfrac \ell 2 + 1; \ell + 2) . \label{eq:combi}
  \end{equation}
  This may be proved by tedious binomial coefficient calculation,
  but a combinatorial proof is possible as follows.
  For any even integer $z \ge 0$, let $\calN(z)$ be the set of all paths
  that start from $(\frac z 2,0)$ and $(0,\frac z 2)$ and merge at some point by moving up and right.
  Take an element $\gamma \in \calN(\ell + 2)$ merging at $(x',y')$ with $x'+y' = f+\ell+2$.
  Let $g = x+y$ be the \emph{least} such that $(x-1,y)$ and $(x,y-1)$ are on~$\gamma$.
  Since the subpaths of $\gamma$ meet at $(x',y')$ for the first time,
  we must have $g \le f + \ell < x' + y'$.
  Cut at $x+y=g$ line, the path~$\gamma$ splits into three components.
  The top-right component defines an element of~$\calN(2)$.
  If we shift the upper left component to the right by~$1$ 
  and shift the lower right component up by~$1$,
  then these shifted components merge at $(x,y)$ but nowhere else since $g$ was chosen to be the least.
  This shift-merged path defines an element of~$\calN(\ell)$.
  Conversely, given $\gamma' \in \calN(\ell)$ and $\gamma'' \in \calN(2)$
  where the sum of total lengths is $2(f+\ell+2)$,
  we can break the tip (the merger point) of $\gamma'$,
  shift the separated components left and down,
  cap them with $\gamma''$, 
  and obtain an element of~$\calN(\ell+2)$.
  This proves~\cref{eq:combi} and in turn the second claim.

  For the third claim we have
  \begin{align}
    Q(2) 
    &= \sum_{k=0}^\infty \sum_{a+b = k} N(a + 1, b + 1; 2) \eta^{2a + 2b} \\
    &= \sum_{k=0}^\infty \eta^{2k} \sum_{a+b = k} \binom{k}{a} \binom{k}{b} - \binom{k}{a+1}\binom{k}{b+1}\nonumber\\
    &= \sum_{k=0}^\infty \eta^{2k} \left[\binom{2k}{k} - \binom{2k}{k+2} \right]\nonumber\\
    &= \sum_{k=0}^\infty \eta^{2k} \frac {1} {k+1} \binom{2k+2}{k} \, . \nonumber
  \end{align}
  Using $(1-4u)^{-1/2} = \sum_{k=0}^\infty \binom{2k}{k} u^k$ for $4u < 1$, we have
  \begin{align}
    \half - \half (1-4\eta^2)^{1/2} = \int_0^{\eta^2} (1-4u)^{- 1/2} \rd u   
    &= \sum_{k=0}^\infty \eta^{2(k+1)} \frac{1}{k+1} \binom{2k}{k} \nonumber\\
    &= \eta^2 + \sum_{k = 1}^\infty \eta^{2(k+1)} \frac{1}{k}\binom{2k}{k-1} \, .
  \end{align}
  Rearranging, we get
  \begin{equation}
    \half - \half (1-4\eta^2)^{1/2} - \eta^2 = \eta^4 \sum_{m=0}^\infty \eta^{2m} \frac{1}{m+1} \binom{2m+2}{m} \, .
  \end{equation}
  Therefore,
  \begin{equation}
    Q(2) = \eta^{-4}\big(\half - \half(1-4\eta^2)^{1/2} - \eta^2 \big) = (1 + q^{-2})^2 \, .
  \end{equation}

  Even if $q = 1$, the calculation for the third claim is still valid since $\binom{2k}{k} \sim 2^{2k}(\pi k)^{-1/2}$
  for large~$k$,
  and we have $Q(2) = 4$ and hence $Q(\ell) = Q(2)^{\ell/2} = 2^\ell$.
  The lemma is true with $q=1$.
\end{proof}

Next, we examine part of domain wall restricted to the region 
where $x + y \le T + \frac \ell 2$.
Define
\begin{align}
  N(z; \ell) &= \sum_{x,y \,:\, x+y=z} N(x,y;\ell) 
  = \begin{cases} 0 & \text{if } z < \ell \\ 1 & \text{if } z = \ell \end{cases} \, ,\nonumber \\
  \calJ(T;\ell) &= \left\{ \gamma|_{x+y \le T + \frac \ell 2} \,\middle|\, \gamma \in \calN(\ell) \text{ merging at $(x,y)$ with $x+y \ge T +\frac \ell 2$ }\right\}\, ,\\
  J(T;\ell) &= \# \calJ(T;\ell) \quad \text{ which is }\begin{cases} \le 4^T & \text{always}\\ = 4^T & \text{if } T < \frac \ell 2 \end{cases}\, . \nonumber
\end{align}
Since the domain walls either remain disjoint or merge,
the purity of the interval at time~$T$ is then
\begin{equation}
  P(T;\ell) = J(T;\ell) \eta^{2T} + \sum_{z < T + \frac \ell 2} N(z;\ell) \eta^{2z - \ell} \, .
  \label{eq:defP}
\end{equation}

\begin{proof}[Proof of~\cref{thm:InfinitePurityBound}]
  (Upper bound)
  The sum over~$z$ in~\cref{eq:defP} is of nonnegative numbers
  and converges to $q^{-\ell}$ as $T \to \infty$ by~\cref{lem:infiniteInfinitePurity}.
  Therefore, $P(T;\ell) - q^{-\ell} \le J(T;\ell)\eta^{2T} \le 4^T \eta^{2T}$.

  (Lower bound)
  Let $\calJ(r;T;\ell) \subset \calJ(T;\ell)$ be the set of all configurations
  such that the two end points on the line $x+y=T$ of the domain wall 
  are separated by $\ell_1$-distance~$r$.
  The distance~$r$ is even and $0 \le r \le 2T + \ell$.
  Let $J(r;T;\ell) = \# \calJ(r;T;\ell)$.
  If we apply~\cref{lem:infiniteInfinitePurity} 
  to the two end points of $\gamma \in \calJ(r;T;\ell)$,
  we see that
  \begin{equation}
    \sum_{z \ge T + \frac \ell 2} N(z;\ell) \eta^{2z - \ell} 
    = 
    \sum_{r = 0}^{2T+\ell} J(r;T;\ell) q^{-r} \eta^{2T} \, .
  \end{equation}
  Since $\calJ(T;\ell) = \bigcup_r \calJ(r;T;\ell)$,
  we see $J(T;\ell)\eta^{2T} = \sum_r J(r;T;\ell) \eta^{2T}$.
  Therefore,
  \begin{equation}
    P(T;\ell) - q^{-\ell} = \sum_{r=0}^{\ell+2T} J(r;T;\ell)(1 - q^{-r})\eta^{2T} 
    \ge J(\ell;T;\ell)(1-q^{-\ell})\eta^{2T} \, .
  \end{equation}
  It remains to show that $J(\ell;T;\ell) \ge \frac{1}{T+1} \binom{2T}{T}$.
  By the reflection trick, it is not difficult to see that
  \begin{align}
    J(\ell;T;\ell) 
    &= \sum_{s=0}^T \binom{T}{s}\binom{T}{s} - \binom{T}{T - s + \frac \ell 2} \binom{T}{s + \frac \ell 2} \nonumber\\
    &= \binom{2T}{T} - \binom{2T}{T+\ell} \\
    &\ge \binom{2T}{T} - \binom{2T}{T+1} = \frac{1}{T+1} \binom{2T}{T} \, . \nonumber \qedhere
  \end{align}
\end{proof}

\Cref{thm:InfinitePurityBound} 
basically means that for large $T$ the interaction between two walkers is weak.
However, this interpretation should not be taken too literally.
Were it not for the interaction, 
we would not have the convergence of the reduced density matrix
to the maximally mixed state.
Note that if $T < \ell /2$, then $P(T;\ell) = ( \frac{2q}{q^2+1} )^{2T}$ 
since there are two independent components of the domain wall~\cite{Nahum2017}.

\section{Tight norm conversions}

The trace distance and fidelity between two density matrices 
are operationally meaningful but analytically difficult.
So, we have resorted to $2$-norms at the expense of a dimension factor 
in our complexity lower bounds.
Such a dimension factor is known to be optimal
if the eigenspectrum of a density matrix is flat.
However, in a rudimentary numerical investigation,
we have observed that the entanglement spectrum of the half-infinite chain
under the brickwork random unitary circuit
is quite far from flat.
Hence, it is natural to wonder if some entropy-derived quantities
can be used in place of the dimension factor in norm conversions.
Here, we remark that such a dimension factor is unavoidable 
even if we constrain a constant number of Renyi $m$-entropy values
with $m \ge 1$.

\begin{remark}
Let us show that the inquality for two normalized density matrices~$\rho,\sigma$
\begin{equation}
    \norm{\rho - \sigma}_1 \le 2 \sqrt{\min(\rank \rho, \rank \sigma)} \norm{\rho - \sigma}_2 \label{eq:OneTwoNormsRank}
\end{equation}
cannot be improved even if we constrain Renyi-($m \ge 1$) entropies.
Let $a = \diag(a_1 \ge a_2 \ge \cdots \ge a_r > 0 = \cdots = 0)$ 
be a diagonal density matrix of rank~$r$ and some large dimension~$d$.
Let $0 < \epsilon < a_r$, and define $b = \diag(b_1 \ge b_2 \ge \cdots b_d)$ by
\begin{equation}
    b_1 = a_1 - \epsilon,\quad 
    b_2 = a_2 - \epsilon, \quad
    \ldots\quad 
    b_r = a_r - \epsilon, \quad 
    b_{r+1} = b_{r+2} = \cdots = b_{d} = \frac{r \epsilon}{d-r} 
\end{equation}
where $d$ is so large that $b_{r} \ge b_{r+1}$.
Then,
\begin{align}
    \norm{a - b}_1^2 &= (2 \epsilon r)^2 \, , \nonumber\\
    \norm{a - b}_2^2 &= r \epsilon^2 + \frac{r^2 \epsilon^2}{d-r}\, ,\\
    \frac{\norm{a-b}_1^2}{\norm{a-b}_2^2} &= 4 r \left( 1 - \frac r d \right) \, . \nonumber
\end{align}
so \eqref{eq:OneTwoNormsRank} is tight in the limit of large~$d$.
On the other hand,
\begin{align}
    \frac 1 r \le \Tr(a^2) < 1, \qquad \Tr(b^2) \to \Tr(a^2) \text{ as } \epsilon \to 0 \, 
\end{align}
where $\Tr(a^2)$ can be anything in that interval.
Therefore, constraining the purity does not improve~\eqref{eq:OneTwoNormsRank}.
Moreover, constraining any fixed number of Renyi-$m$ entropies for $m=1,2,3,\ldots$
does not improve~\eqref{eq:OneTwoNormsRank}.
\end{remark}

\begin{remark}
Let us show the opposite inequality, which is more important for our purposes,
\begin{equation}
    \big( \norm{\sqrt{\rho}\sqrt{\sigma}}_1 \big)^2 \le \min(\rank \rho,\rank \sigma) \Tr(\rho \sigma)
\end{equation}
cannot be improved even if we constrain Renyi-($m \ge 1$) entropies.
Let $\rho = \diag(1-\epsilon, \frac \epsilon d, \cdots, \frac \epsilon d, 0)$ and
$\sigma=\diag(0, \frac \epsilon d, \cdots, \frac \epsilon d, 1- \epsilon)$
be diagonal density matrices of dimension~$d+2$.
We see that
\begin{align}
    \frac{\norm{\sqrt{\rho}\sqrt{\sigma}}_1^2}{\Tr(\rho \sigma)}
    =
    d
\end{align}
for any~$\epsilon > 0$.
Observe that the Renyi-$m$ entropy for any $m \ge 1$ converges to zero as $\epsilon \to 0^+$.
\end{remark}

\section{Shallow unitary designs as a model of chaos}\label{app:shallowDesigns}

A unitary $k$-design 
(a probabilistic ensemble of unitaries that is close to the Haar random one) 
is sometimes taken as a model of chaotic dynamics~\cite{Brandao2019}.
Here we explain significant differences 
between ``shallow'' unitary designs~\cite{LaRacuente2024,Schuster2024,Cui2025}
and the brickwork random circuits in the phenomenology of subregion complexity.
We point out that 
the complexity of the reduced density matrix of an interval of length~$\ell \ll L$ (say $\ell = L/100$)
at depth~$T \gg \ell$, is of order~$T$ 
under the $\epsilon$-approximate unitary $k$-designs constructed in~\cite{Schuster2024}.
Importantly, $T$ here can be exponentially large in~$\ell$.
This contrasts sharply to the brickwork random circuit,
under which the complexity at such late times is zero 
since the state is the maximally mixed state 
up to a trace-distance error that is exponentially small in~$\ell$.

\begin{figure}[h]
    \begin{subfigure}{0.48\textwidth}
        \includegraphics[width=\textwidth, trim={0mm 30mm 0mm 45mm}, clip]{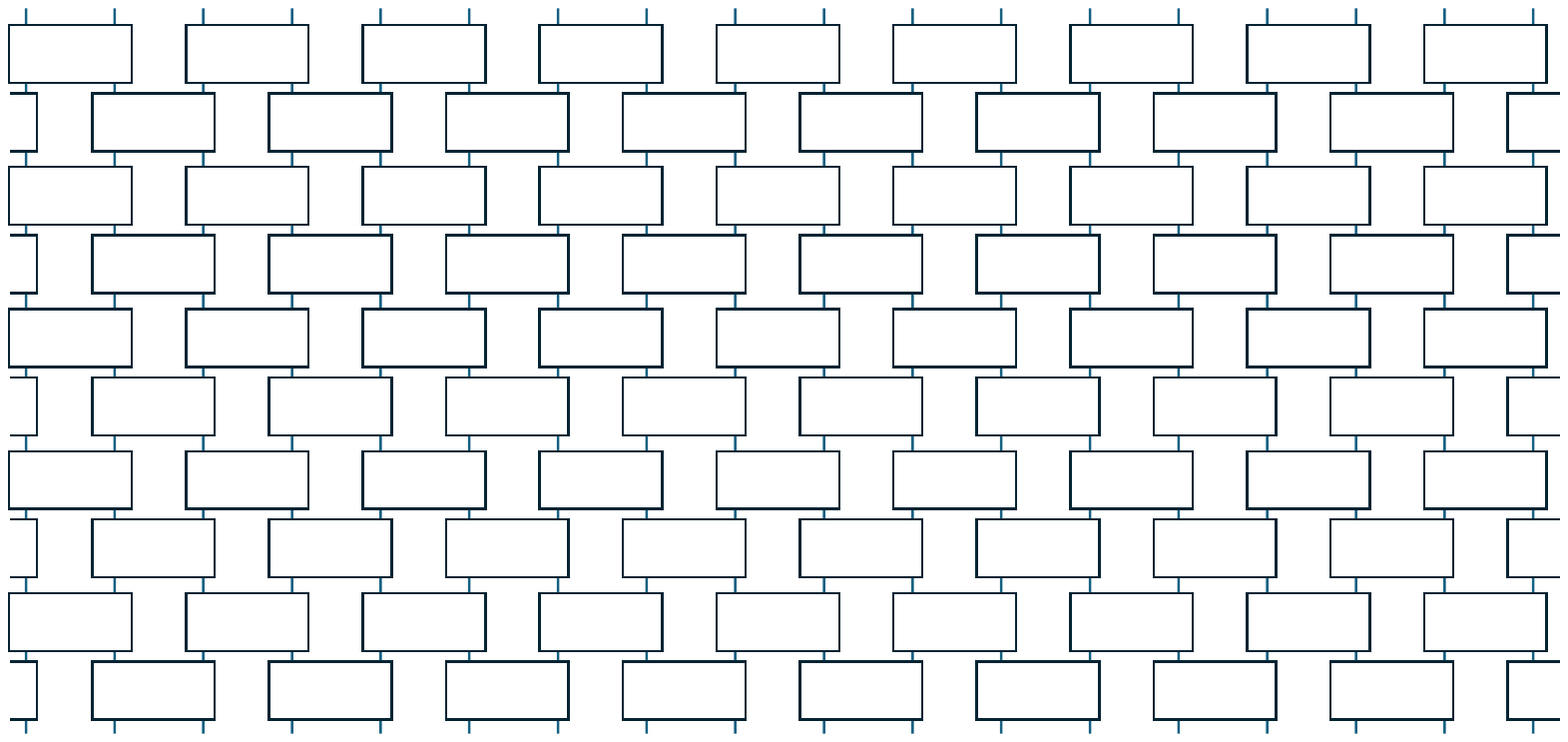}
        \caption{Brickwork circuit in $1+1$d}
    \end{subfigure}
    \begin{subfigure}{0.48\textwidth}
        \includegraphics[width=\textwidth, trim={0mm 30mm 0mm 30mm}, clip]{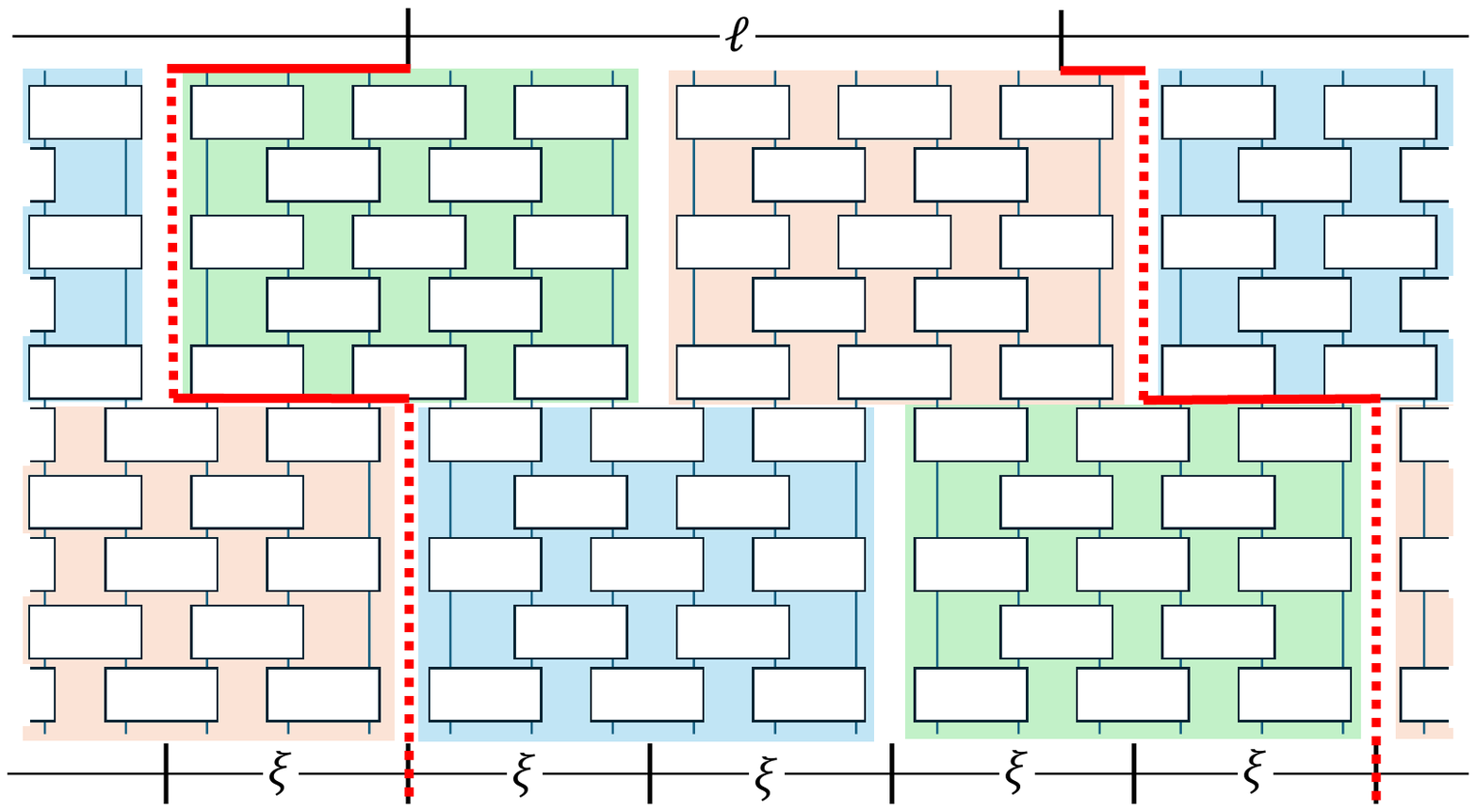}
        \caption{Same as (a) but with the gates straddling between the supergates deleted.}
    \end{subfigure}
    \caption{With gates along $2\xi$-separated columns removed, 
    we still have an approximate unitary $k$-design
    where $k$ is proportional to the depth.
    Starting with a pure product state, 
    the state of an interval under (b) has much smaller entropy than that under (a).}
    \label{fig:ShallowDesign}
\end{figure}

Consider a $1+1$d circuit~$U$ of supergates with depth~$2$
over $q^\xi$-dimensional superqudits,
where every supergate acts on neighboring two superqudits
and is a random unitary drawn from an $\epsilon_0$-approximate unitary $k$-design.
For sufficiently large~$\xi$, the circuit~$U$ is an approximate unitary $k$-design,
where the approximation error~$\epsilon$ is bounded above by iteratively applying the following.
\begin{lemma}[Lemma~8 of~\cite{Schuster2024}]
    Let $U_{AB}$ and $U_{BC}$ be independent approximate unitary $k$-designs
    of errors~$\epsilon_{AB}$ and~$\epsilon_{BC}$
    acting on subsystems~$AB$ and $BC$, respectively.
    Assume that $k^2 d_B^{-1} \le \tfrac 1 2$ where $d_B$ is the Hilbert space dimension of~$B$.
    Then, the product $U_{AB} U_{BC}$ is an $\epsilon_{ABC}$-approximate unitary $k$-design
    where
    \begin{equation}
        1+\epsilon_{ABC} 
        \le (1+\epsilon_{AB})(1+\epsilon_{BC})
        e^{\frac{k^2}{2 d_B} }
        \left( 1 + \frac{k^2}{d_{ABC}} \right)
        \left( 1 - \frac{k^2}{2 d_{AB}} \right)^{-1}
        \left( 1 - \frac{k^2}{2 d_{BC}} \right)^{-1} \, . \label{eq:epsInSUD}
    \end{equation}
\end{lemma}
\noindent
Let $x = k^2 / d_B \le \frac 1 2$. 
For any positive real~$z$, we use the notation $z' = \ln(1+z)$.
Since $d_{AB} \ge d_B$ and $d_{BC} \ge d_B$,
all the $k$-involving factors are at most~$1+x$.
Taking log of~\eqref{eq:epsInSUD} we have that
\begin{equation}
    \epsilon'_{ABC} + 4 x' \le (\epsilon'_{AB} + 4 x') + (\epsilon'_{BC} + 4 x') \, . \label{eq:epsInSUD2}
\end{equation}
Applying~\eqref{eq:epsInSUD2} iteratively, once for each supergate,
we see that
the full circuit~$U$ of $m$ supergates is an $\epsilon$-approximate unitary $k$-design
where
$
    \epsilon' \le m \epsilon'_0 + 4 m x'
$ if $k^2 q^{-\xi} \le \frac 1 2 $.
Use $m \le L / \xi$ and $z' \le z$,
and choose
$\xi = c \ell = c^2 L $ with $c = 1/100$.
Then, the approximation error~$\epsilon$ of~$U$ as a unitary $k$-design satisfies
\begin{equation}
    \ln( 1 + \epsilon) \le c^{-2} \epsilon_0 + 4 c^{-2} k^2 q^{-c \ell}\, .
\end{equation}
Each random supergate can be chosen to be a brickwork circuit 
consisting of elements of~$\U(q^2)$ in $T/2$ layers.
So, the full circuit~$U$ has depth~$T$.
The result of~\cite{Chen2024} implies that each supergate is an $\epsilon_0$-approximate
unitary $k_0$-design for some $k_0 \ge T / \poly(\xi)$ with $\epsilon_0 = c^3$.
Setting $k = T / \poly(\xi)$,
we have $\epsilon' \le c + q^{-c \ell} T^2 / \poly(c \ell)$,
which is less than~$0.1$ if $T < T_{max} = \tilde{\mathcal{O}} (q^{c \ell / 2})$.%
\footnote{In~\cite{Schuster2024},
    the overlap size $\xi$ was taken to be $\mathcal O(\log_q(L k/\epsilon))$
    to minimize~$T$ given~$k$; with this choice of~$\xi$ for a given~$k$, 
    the overall depth~$T$ can be chosen as~$\tilde{\mathcal O}( k \log(L k / \epsilon) )$.
    This is why it is a ``shallow'' unitary design.
}
Note that the overall circuit~$U$ is nothing but a brickwork circuit 
with some $\U(q^2)$ gates \emph{deleted} along time-like columns separated by~$2\xi$.
See~\cref{fig:ShallowDesign}.

Due to the deletion,
the rank~$r$ of the reduced density matrix~$\sigma$ on any interval of length~$\ell$
after applying~$U$ to a product pure state 
is bounded as
\begin{equation}
    r \le q^{4\xi} \ll q^\ell \, .
\end{equation}
(Note that $\norm{\sigma - q^{-\ell} \one_{q^\ell}}_1 > \frac 1 2$ since $r \le \frac 1 2 q^\ell$.)
Then, we follow the same line of reasoning as in~\cref{sec:ComplexityGrowth} for a complexity lower-bound.
In particular, \cref{eq:AvgFidelityKBound} becomes
\begin{equation}
    \avg_\sigma F(\rho,\sigma)^{2k} \le (1+\epsilon) r^k q^{-k \ell} k!
    \le 2 \left( k  q^{-(1-4c)\ell} \right)^k . \label{eq:shallowUnitaryFidelity}
\end{equation}
Setting~$k = T / \poly(c \ell)$, we conclude that typical~$\sigma$ has complexity
\begin{equation}
    \mathfrak C(\sigma) = \Omega(T / \poly(c \ell,q)) \text{ for } T \le T_{max} = q^{\mathcal O (\ell)}. \label{eq:compLBShallowDesign}
\end{equation}

So the rank of $\sigma$ is small, and its complexity can be large. These results contrast with the case where $U$ is Haar random --- in that case $\sigma$ would be close to maximally mixed, with low complexity. How are these opposite conclusions compatible with the fact that the supergate circuit generates an approximate $k$-design? To temporarily sharpen the puzzle, note that for the Haar random case, these facts can be deduced by computing the purity, which is only quadratic in the density operator:
\begin{align}
   P(U) = \Tr[ (\Tr_{L-\ell} U 0^{\otimes L} U^\dagger)^2 ]
    &= 
    \Tr[
        \SWAP_{A,A'} (U 0^{\otimes L} U^\dagger) \otimes (U 0^{\otimes L} U^\dagger)
    ]\, , \label{eq:swapexp}\\
    \avg_{U \sim \U(AB)} ~P(U) 
    &= \frac{d_A + d_B}{d_A d_B + 1} 
    \le \frac 1 {d_A}\left( 1 + \frac{2d_A}{d_B} \right)  \nonumber
\end{align}
where $d_A = q^\ell$ and $d_B = q^{L-\ell}$.
Converting this into a bound on the trace distance,
we have
\begin{equation}
    \norm*{\rho_A - \frac{\one_A}{d_A}}_1 ^2 \le d_A \norm*{\rho_A - \frac{\one_A}{d_A} }_2^2 
    = P(U) \, d_A - 1 \le \frac{2 d_A}{d_B} = 2 q^{-(L - 2 \ell)} \, .
\end{equation}
Since this calculation only involves $(U \otimes U^\dagger)^{\otimes 2}$,
it suffices for $U$ to be an exact unitary ($k=2$)-design.

However, with $\epsilon$-approximate unitary designs, the error parameter~$\epsilon$ 
becomes very important.
Recall that in 
the defining inequality for the relative error~$\epsilon$ in \cref{eq:approxDesignEps},
the observable~$M$ is required to be positive semi-definite,
but $\SWAP_{A,A'}$ has negative eigenvalues.
One can decompose $\SWAP$ into positive and negative parts,
use \cref{eq:approxDesignEps},
and combine the results,
but this gives a completely meaningless result unless $\epsilon$ is exponentially small in~$\ell$.
(In view of~\eqref{eq:compLBShallowDesign}, 
such a small $\epsilon$ will require $c = \xi / \ell$ to be ``big.'')
In fact, a bound on trace distance with a constant error~$\epsilon$ for an approximate unitary design
must \emph{not} imply that $\rho_A$ is close to the maximally mixed state
since we have given a lower bound on the complexity using a special approximate unitary design.
This is a drastic instance of a sign problem.
In the brickwork random circuit, 
we have used separate calculation (\cref{thm:InfinitePurityBound}) of purity,
which is not related to the unitary $2$-design property.

Note that the random circuits of~\cite{LaRacuente2024,Schuster2024,Cui2025}
do not define a simple Markovian random walk on the unitary group.
That is,
if we have a time sequence of unitaries $U(T)$ of depth~$T = 1,2,3,\ldots$,
constructed by the prescription of those references,
then $U(T+1)U(T)^\dagger$ is not a quantum circuit of constant depth.
Hence, it seems absurd to use these constructions to model time dynamics.
If we demand the simple Markovian property,
{\it i.e.}, 
the distribution of $U(T+1)$ must be the convolution 
of some distribution of constant depth quantum circuits
and the distribution of $U(T)$,
then, perhaps, the complexity of any smaller-than-half subsystem 
would stop growing at time~$\Theta(\ell)$.

\section{Random matrix integral analysis of random projectors}\label{app:rmt}
Let $\pi_i$ be projection operators of rank $r$ in a Hilbert space of dimension $d$, and let $\sigma_i = \pi_i/r$ be the corresponding density matrices. How many $\sigma_i$ can one pick in such a way that all pairs are approximately orthogonal? Here, the measure of orthogonality we will use is the fidelity:
\begin{equation}
    F(\sigma_i,\sigma_j) \equiv \tr \left(\sqrt{\sqrt{\sigma_i}\sigma_j\sqrt{\sigma_i}}\right) = \frac{1}{r}\tr(\sqrt{\pi_i\pi_j\pi_i}).
\end{equation}
The requirement of approximate pairwise orthogonality is
\begin{equation}\label{rootfidbound}
    F(\sigma_i,\sigma_j) < \epsilon \hspace{20pt} \text{ for all} \hspace{20pt} i \neq j,
\end{equation}
where $\epsilon$ is a small tolerance parameter that is fixed in the limit of large $r,d$.

As a strategy to generate the projectors, let's pick randomly
\begin{equation}
    \pi_i = U_i \cdot \text{diag}(\underbrace{1,\dots 1}_{r},\underbrace{0,\dots 0}_{d-r})\cdot U_i^\dagger
\end{equation}
where $U_i$ are independent Haar random $d\times d$ unitaries. As we will show below, for $r \ll d$, independently-drawn projectors of this type will typically have fidelity $F \approx \frac{8}{3\pi} \sqrt{r/d}\ll 1$. We will also calculate the probability $\Pr(\epsilon)$ that the fidelity is anomalously high, of order $\epsilon$. If one tries to form an ensemble with $N$ pairwise-orthogonal states, then as long as $\Pr(\epsilon)N < 1$, it will be possible to add another state (by picking a new state randomly, and if necessary repeating), so it will be possible to construct an ensemble with up to $N \sim 1/\Pr(\epsilon)$ states.

To compute $\Pr(\epsilon)$ we can use results from \cite{Collins2004} for the probability distribution for the $r$ nonzero eigenvalues of 
\begin{equation}
    \pi_1\pi_2\pi_1.
\end{equation}
This distribution is (theorem 2.2 of \cite{Collins2004})
\begin{equation}\label{collins}
    \Pr(\{\lambda_1,\dots,\lambda_d\}) \propto \prod_{i=1}^r (1-\lambda_i)^{d-2r}\prod_{1\le i<j\le r}(\lambda_i-\lambda_j)^2, \hspace{20pt} 0\le\lambda_i \le 1.
\end{equation}
To compute the fidelity, we would like to estimate this probability subject to a constraint on the sum of the square roots of the eigenvalues.

In the limit of large $r,d$ with fixed $w\equiv r/d < 1/2$, the discrete set of eigenvalues can be approximated as a smooth density $\rho(\lambda)$. The probability can be expressed in terms of this density
\begin{align}
   \log \Pr[\rho(\lambda)] = d^2 \bigg\{\int \mathrm{d}\lambda \mathrm{d}\lambda' \log|\lambda-\lambda'| + (\tfrac{1}{w}-2)\int \mathrm{d}\lambda \rho(\lambda)\log(1-\lambda)\notag\\
   + \alpha_1 \left[1-\int \mathrm{d}\lambda \rho(\lambda)\right]+ \alpha_2\left[F-\int \mathrm{d}\lambda \rho(\lambda)\sqrt{\lambda}\right]\bigg\}.\label{logP}
\end{align}
The first line is the translation of (\ref{collins}), and the second line are Lagrange multipliers enforce the constraint that the distribution should be normalized and have fidelity equal to $F$. This is a standard Hermitian matrix integral problem with a constraint $\lambda > 0$ and potential
\begin{align}
    V(\lambda) = -(\tfrac{1}{w}-2)\log(1-\lambda) + \alpha_1 + \alpha_2\sqrt{\lambda}.
\end{align}
In principle, to compute the probability of a given fidelity, one should find the saddle point $\rho(\lambda)$ as a function of $\alpha_2$, adjust $\alpha_2$ to give the desired fidelity, and then evaluate (\ref{logP}). This will give 
\begin{equation}\label{mystery}
    \log \Pr(F) = -r^2\cdot (\text{some function of }\frac{r}{d}, F).
\end{equation}
In practice, one can use shortcuts to compute the mean and variance of this distribution, which will be enough for our purposes.

The mean can be computed from the saddle point density with no constraint on the fidelity, $\alpha_2 = 0$. This is
\begin{equation}
    \rho(\lambda) = \frac{1}{2\pi w}\frac{\sqrt{a-\lambda}}{\sqrt{\lambda}(1-\lambda)}, \hspace{20pt} 0\le \lambda \le a, \hspace{20pt} a = 4w(1-w). 
\end{equation}
which gives
\begin{align}
    \mathbb{E}(F) &= \int \mathrm{d}\lambda \rho(\lambda) \sqrt{\lambda} \\ &= \frac{2}{\pi}\left(\sqrt{\frac{1-w}{w}} - \frac{1-2w}{w}\text{arcsin}(\sqrt{w})\right)\\ &\approx \frac{8}{3\pi}\sqrt{\frac{r}{d}},\label{fidexpform}
\end{align}
where the last line is an approximate expression for $r\ll d$.

To compute the variance for large $r$, one can use the general formula for the connected correlation function of resolvents, see e.g.~(4.32) of \cite{Stanford:2019vob}, with $a_- = 0, \beta = 2,a_+ = 4w(1-w)$. By expanding that formula for large $x_1,x_2$ and studying the term of order $x_1^{-(1+n)}x_2^{-(1+n)}$, one finds 
\begin{equation}
\mathbb{E}\left\{\Big(\tr\left((\pi_1\pi_2\pi_1)^n\right)\Big)^2\right\}-\left(\mathbb{E}\Big\{\tr\left((\pi_1\pi_2\pi_1)^n\right)\Big\}\right)^2 = \frac{\Gamma(n-\frac{1}{2})^2}{4\pi\Gamma(n)\Gamma(n-1)}a^{2n}.
\end{equation}
The variance of the fidelity is $1/r^2$ times the continuation to $n = 1/2$ of this expression, which is
\begin{equation}
\mathbb{E}(F^2) - \mathbb{E}(F)^2 = \frac{2w(1-w)}{\pi^2 r^2}.
\end{equation}
So we have computed enough of the mystery function in (\ref{mystery}) to know
\begin{equation}\label{fidvarform}
    \log \Pr(F) = -r^2 \left[\frac{\pi^2}{4\frac{r}{d}(1-\frac{r}{d})}\big(F-\mathbb{E}(F)\big)^2  + \dots\right]
\end{equation}
where the dots are of order $(F-\mathbb{E}(F))^3$ and higher.

\bibliographystyle{alphaurl}
\bibliography{comp-refs}


\end{document}